%% file: mainv1.tex
\documentclass{article}

\usepackage{fullpage}
\usepackage{amsthm}
\usepackage{natbib}

\usepackage[utf8]{inputenc} 
\usepackage[T1]{fontenc}    
\usepackage{hyperref}      
\hypersetup{
  colorlinks=true,
  citecolor=blue,
  linkcolor=black,
  urlcolor=black
}
\usepackage{url}            
\usepackage{booktabs}       
\usepackage{amsfonts}       
\usepackage{nicefrac}       
\usepackage{wrapfig}
\usepackage{microtype}      
\usepackage{xcolor}         
\usepackage{nicefrac}       
\usepackage{caption}        
\input{aux/math_commands}

\usepackage[algo2e,ruled,vlined,linesnumbered]{algorithm2e}
\usepackage{float}
\newtheorem{lemma}{Lemma}
\newtheorem{theorem}{Theorem}
\newtheorem{assumption}{Assumption}
\newtheorem{proposition}{Proposition}
\newtheorem{definition}{Definition}

\usepackage{graphicx}
\usepackage{cleveref}
\usepackage{subfigure}
\usepackage{enumitem}
\usepackage[none]{hyphenat}
\usepackage{amssymb}
\usepackage{authblk}

\title{Learning to Price with Persuasion}

\author[1]{Maria-Florina Balcan}
\author[1]{Tejas Pagare}
\author[1]{Karan Singh}

\affil[1]{Carnegie Mellon University, Pittsburgh, USA}

\affil[ ]{\texttt{\href{mailto:ninamf@cs.cmu.edu}{ninamf@cs.cmu.edu}, \href{mailto:tpagare@andrew.cmu.edu}{tpagare@andrew.cmu.edu},  \href{mailto:karansingh@cmu.edu}{karansingh@cmu.edu}}}

\date{}
\begin{document}
\maketitle
\allowdisplaybreaks
\input{src/abstract}

\input{src/intro}
\input{src/relwork}
\input{src/preliminaries}
\input{src/samplecomplexity}

\input{src/dynamic_programming}

\input{src/query}

\input{src/robustregret}

\input{src/conclusion}

\bibliography{aux/New_USE_THIS}
\bibliographystyle{apalike}

\newpage
\appendix
\input{src/appendix}


\end{document}

%% file: aux/math_commands.tex
\usepackage{amsmath,amsfonts,bm}

\def\eqref#1{equation~\ref{#1}}

\def\1{\bm{1}}

\usepackage{float}

\def\vmax{{\overline{v}}}

\def\qmin{{\underline{q}}}
\def\qmax{{\overline{q}}}

\DeclareMathAlphabet{\mathsfit}{\encodingdefault}{\sfdefault}{m}{sl}
\SetMathAlphabet{\mathsfit}{bold}{\encodingdefault}{\sfdefault}{bx}{n}

\newcommand{\E}{\mathbb{E}}

\DeclareMathOperator*{\argmax}{arg\,max}

%% file: src/abstract.tex
\begin{abstract}
Motivated by modern marketplaces, where the platform or the seller routinely gathers detailed user profiles, we study a novel learning theoretic model that simultaneously involves information and mechanism design. Specifically, we consider the economic setting recently introduced by \citet{bergemann2022screening}, where in addition to the menu of quality-price pairs, the seller offers information on the value of the match between product quality and buyer's taste via a signaling scheme. We relax the assumption that the seller knows the buyers' belief about the distribution of tastes and study the sample requirements of designing a revenue maximizing scheme. We consider both the batch setting where we have access to data from a set of i.i.d. buyers and an online demand query model where we observe the buyers' behaviors to seller's schemes. Despite the apparent non-convexity of the problem, we also give the first FPTAS to compute a scheme that maximizes the revenue within an arbitrarily small additive loss, which was left open by \citet{bergemann2022screening}. Overall, this brings a new learning perspective in asymmetric economic settings where buyers and sellers know different types of information.
\end{abstract}

%% file: src/intro.tex
\section{Introduction}
A cornerstone of microeconomic theory is information asymmetry. Economic agents participating in a market might have latent preferences (or motives) invisible to the market maker. The foundational works of \citet{spence1973job, akerlof1970market} detail signaling mechanisms by which participants can partially overcome this information asymmetry, and market failures that result from the unmitigated presence of these circumstances. However, in recent years, the presence of monolithic algorithmic marketplaces has altered this picture. Such marketplaces routinely gather detailed user profiles, consisting of both demographic data and the history of past purchases. For example, recently, a shopper was surprised to discover that Kroger, a not particularly tech-savy retail giant in the US, had a 62-page shopping profile on them, including details like pet ownership and travel history \citep{kroger}. As a result, such digital platforms can forecast the personalized value a product or service holds for a given buyer much more accurately than individual buyers themselves. As another example, Google, as an advertising marketplace serving a trillion impressions every month, can arguably estimate the conversion rate associated with displaying a specific advertisement in a specific spot far more accurately than individual publishers.

We consider one of the most foundational problems, namely, monopolistic screening \citep{MUSSA1978301, maskin1984monopoly}, albeit in the presence of such an informationally advantaged seller. In the classical problem, a monopolist aims to maximize the revenue gathered from a pool of buyers who have a certain distribution of latent tastes, by offering a menu of quality-price pairs. Products of higher quality also cost the monopolist more. By offering a menu of varied qualities at different price points, the seller can simultaneously appeal to disparate buyers whose utility for a given product is a function of both their latent taste and the product's quality. In the new-age model, proposed by \cite{bergemann2022screening} (also see \cite{bergemann2026screening}), the key distinction is that the seller can forecast the utility that a buyer with a certain taste associates with different product qualities, far better than the buyers themselves. A notable constraint on the seller is that in spite of knowing the buyer's utility precisely she can not engage in perfect price discrimination (driving the consumer surplus to zero), because existing regulations exclude explicit price discrimination in most markets \citep{maryland}. Thus, in addition to designing a {\em public} menu of quality-price pairs, the seller commits to releasing partial information about the value of buyer-quality matches via a signaling scheme of their choice allowing an indirect price discrimination. In practice, one can think of such a signal as a personalized recommendation to the buyer indicating what product (quality) she should purchase (or in a softer manner through search rankings). Recommendation algorithms in this context play the role of a commitment device. For example, Amazon labels certain products, which are neither the ones with the highest rating nor the highest number of reviews, as {\em Amazon's choice}. Such personalized recommendations can be viewed as signals to the buyer.

\begin{figure}[t]
    \centering
    \begin{tabular}{| c | c | c |} 
\hline
  {\bf Model} & {\bf Value distribution} & {\bf Sample Complexity / Regret} \\ [0.5ex] 
 \hline\hline
 Value samples & Arbitrary & $\widetilde{\mathcal{O}}(1/\varepsilon^3)$ \\ 
 \hline\hline
 Demand queries & Discrete & $\mathcal{O}(1)$ \\
 \hline
 & Bounded pdf & $\widetilde{\mathcal{O}}(1/\varepsilon^2)$ \\
 \hline
 & Lipschitz pdf & $\widetilde{\mathcal{O}}(1/\varepsilon)$ \\
 \hline
  & Analytic pdf & $\mathcal{O}(\textrm{polylog } 1/\varepsilon)$ \\
 \hline\hline
 Joint learning & Arbitrary &  $\widetilde{\mathcal{O}}(T^{3/4})$ \\
 \hline
\end{tabular} 
    \caption{Summary of our results. For value samples and demand queries, we list the sample requirements to produce an $\varepsilon$-optimal menu along with the associated signaling scheme. The last row lists the regret in a $T$-horizon game.}
    \label{tab:1}
\end{figure}

To design a revenue maximizing menu and signaling scheme in the above context, it is typically assumed that the seller knows the precise distribution of the buyer's tastes. This is implausible given the diversity of products sold on modern marketplaces, and in this work, we relax this assumption. Instead, we consider explicit learning algorithms that dictate how such knowledge may be acquired via samples and past interactions, with explicit bounds on sample and compute requirements. Concretely, we make the following contributions (also partly summarized in \Cref{tab:1}). 
\begin{enumerate}[font=\bfseries,itemindent=1.3em, leftmargin=0em, itemsep=0em]
    \item We give a learning algorithm that when given $\widetilde{\mathcal{O}}(1/\varepsilon^3)$ i.i.d. samples from the buyers' taste (or as we will call it going forward, {\em value}) distribution, produces a signaling scheme and menu that achieves a revenue within $\varepsilon$ of the maximum.
    \item The problem of jointly designing a menu and a signaling scheme is non-convex, as noted in \citet{bergemann2022screening}. Despite this, we give the first FPTAS that computes a solution in polynomial time with revenue within an arbitrarily small additive loss of the optimum. 
    \item We also study a {\em demand query} model, where the seller can observe how the buyers behave in presence of the menus and signaling schemes she designs. In this interactive setting, we obtain a constant sample complexity for discrete value distributions and improve to $\widetilde{\mathcal{O}}(1/\varepsilon^2)$ samples for continuous distributions with bounded pdfs. Further smoothness assumptions result in smaller sample requirements; for example, polylogarithmic for analytic pdfs.
    \item Finally, we give regret upper bounds for a model in which both the seller and the buyer population jointly learn in an online setting based on the past realizations of values. In this setting, the seller has to design schemes that are somewhat robust to the buyers' beliefs, which are incompletely specified.
\end{enumerate}

%% file: src/relwork.tex
\subsection{Related Work}
In economics, the nonlinear pricing problem introduced in \cite{MUSSA1978301} and further studied in \cite{maskin1984monopoly} characterizes the optimal quality, price menus for a revenue maximizing monopolist seller facing privately informed random buyer or a population of buyer with known value distribution. See \cite{wilson1993nonlinear} for a broader survey and \cite{bolton2004contract} for the simplified analysis for computing optimal menu.The information design literature in which a informationally advantaged sender influences receiver via communication was initially popularized by a cheap talk model \cite{crawford1982strategic}. This was then formalized with a commitment power to the sender on distribution of messages as Bayesian Persuasion in \cite{kamenica2011bayesian}. See \cite{bergemann2019information,kamenica2019bayesian} for a survey treatment. This was then translated to the linear bayesian persuasion setting \cite{gentzkow2016rothschild,dworczak2019simple} where the senders' utility is solely characterized by the expected posterior value. We build on \citet{bergemann2026screening}, who prove that the seller-optimal joint design always takes a monotone partitional form with finitely many signals and menu items. See also \cite{bergemann2026information} for more developments on the joint mechanism and information design problem.

The intersection of machine learning and economics is now a flourishing area of research. The use of sample complexity theory for mechanism design was initiated by~\cite{balcan2005mechanism,balcan2008reducing}. It was later studied by many others including \cite{cole2014sample,roughgarden2016ironing} for single-item auction and \cite{morgenstern2016learning,balcan2016sample,mohri2016learning} for parametrized auction classes. Online learning with unknown demand distribution for posted pricing was studied in \cite{kleinberg2003value}, see also \cite{balcan2024new} and \cite{balcan2025generalization} for other pricing mechanisms. Online learning in Bayesian Persuasion, where the sender faces a receiver with unknown type, which characterizes the receiver's utility function, was initiated in \cite{castiglioni2020online} and further developed in \cite{babichenko2022regret,zu2025learning,castiglioni2021multi}. Recent developments include \cite{bacchiocchi2024online} which further considers that the sender does not know the common prior, \cite{zu2025learning} studies the sender interaction with a population of receivers who also does not know the common prior. See also \cite{balcan2026nearlyoptimal, balcan2026learning} who consider learning in the presence of additional side-information which is public to both the sender and receiver.

A notable work on the learning-theoretic interface of the joint information design and pricing problem is \citet{agrawal2023dynamic}, which studies a variation of the posted pricing problem in which the seller reveals information about the uncertain quality of a good. However, the mechanism formats, the posted pricing in the former, and the menu pricing for us, are fundamentally incomparable. For instance, the ratio of optimal revenues for posted pricing for a single quality product vs. menu pricing can be vanishingly small. In addition, the former assumes that the seller knows the (exogenous) quality distribution but not the buyers' demand distribution, whereas in our case, the seller can produce any quality and influence the demand/type of the buyer through carefully constructed signals. 

%% file: src/preliminaries.tex
\section{Problem Setup and Preliminaries}
\label{section:model}
We consider the setting from \citet{bergemann2022screening}, where a monopolist seller can manufacture goods of arbitrary (positive) quality $q$ in the range $\mathcal{Q}=[\underline{q},\overline{q}]$ at a cost $c(q)$, which is assumed to be convex and non-decreasing. The buyer's value $V$ is distributed according to a common prior on $\mathcal{V}=[0,1]$ with the CDF $F(v):=P(V\leq v)$. Following \citet{MUSSA1978301}, a buyer with value $v$ gains a utility of $v\cdot q-p$ upon purchasing a product of quality $q$ at price $p$. Buyers in this setup are Bayesian utility-maximizing agents, who are ex ante unaware of their own value in contrast to \citet{MUSSA1978301}.


The  seller commits to a public menu and a signaling scheme; the latter is sometimes termed an information structure or a Blackwell experiment \citep{blackwell1953equivalent}. The public menu $\mathcal{M}=\{(q_i,p_i): i\in I\}$ is composed of a collection of quality-price tuples for some index set $I$, with qualities restricted to $\mathcal{Q}$ and prices at most $\overline{q}$ so that buyers still find the menu attractive. By an argument mirroring the revelation principle, sometimes termed an {\em obedience} argument, we can generically assume that the signal space $\mathcal{S}$ as being in a one-to-one correspondence with the menu and, therefore, just as numerous. Intuitively, each menu item can be labeled by the signal for which it is the best response -- if there is no such signal, it can be removed altogether -- and signals inducing the same menu item as the best response can be merged. Formally, the signaling scheme $\mathcal{I}:=(\mathcal{S},\{\pi(\cdot|v)\}_{v\in \mathcal{V}})$  encodes a signal space $\mathcal{S}\subseteq \mathcal{Q}$ and a conditional distribution over signals for every possible value $\pi(\cdot|v)\in\Delta(\mathcal{S})$. To fully specify the timeline, upon seeing a buyer with value $v$, the seller samples a signal $s\sim \pi(\cdot | v)$, as promised upon commitment. The buyer, upon observing a signal $s$, forms a posterior mean $\bar{v}_{\pi,F}(s) = [\int_{0}^{1}v\pi(s|v)\mathrm{d}F(v)]/[\int_{0}^{1}\pi(s|v)\mathrm{d}F(v)]$, and rationally chooses a quality $(\mathsf{q}(s),\mathsf{p}(s)) \in \arg\max_{(q,p)\in \mathcal{M}}\bar{v}_{\pi,F}(s)\cdot q-p$, breaking ties in the favor of the seller. This gives the seller a revenue of $\mathsf{p}(s)-c(\mathsf{q}(s))$ which is the profit minus the production cost. The seller thus aims to find an incentive-compatible direct menu $\{(\mathsf{q}(s),\mathsf{p}(s)):s\in\mathcal{S}\}$ and a signaling scheme $\mathcal{I}$ to maximize the expected revenue which corresponds to the following program where the first constraints enforce individual rationality (\texttt{IR}) and the second constraints enforce incentive compatibility (\texttt{IC}) which ensures that menu item $(\mathsf{q}(s),\mathsf{p}(s))$ is a best-response for the buyer upon receiving signal $s$.
\begin{align}
    \max_{\substack{\mathcal{S}\subseteq \mathcal{Q},\ \{(\mathsf{q}(s), \mathsf{p}(s)):s\in \mathcal{S}\}\\\{\pi(\cdot|v)\in \Delta(\mathcal{S}):v\in \mathcal{V}\}}}  &\sum_{s\in\mathcal{S}} \int_{0}^{1} (\mathsf{p}(s)-c(\mathsf{q}(s))) \pi(s|v) \mathrm{d}F(v) \label{eq:opt1}\\ 
   \text{subject to }& \bar{v}_{\pi,F}(s)\cdot \mathsf{q}(s)- \mathsf{p}(s) \geq 0\ \forall s\in\mathcal{S}\quad (\texttt{IR})\quad \text{and} \nonumber \\
    &\bar{v}_{\pi,F}(s)\cdot \mathsf{q}(s)- \mathsf{p}(s)  \geq \bar{v}_{\pi,F}(s)\cdot\mathsf{q}(s')- \mathsf{p}(s')\ \forall s,s'\in\mathcal{S}\quad (\texttt{IC}).\nonumber
\end{align}

A key result from \citet{bergemann2022screening} that we will use states that the resulting menu has only a finite number of items, even if the underlying value distribution is continuous. In contrast, in \citet{MUSSA1978301}, where there is no information provision, this number can be unbounded.


\begin{definition}\label{def:1}
    A signaling scheme $\mathcal{I}:=(\mathcal{S},\{\pi(\cdot|v)\}_{v\in \mathcal{V}})$ belongs to the set of monotone partitional information structures $\mathcal{I}_{\mathsf{part}}$, if there exists an $I\in \mathbb{Z}_{>0}$ and $0=t_0< t_1<\cdots< t_I=1$ such that upon receiving any signal $i$ each buyer only knows that their value corresponds to the quantile in $[t_{i-1}, t_{i})$ with respect to the distribution of values. 
\end{definition}

\begin{theorem}[\cite{bergemann2026screening}]
The optimal menu has as most $K:= \qmax/\qmin$ items. Furthermore, the optimal information structure is monotone partitional with at most $K$ distinct signals.
\end{theorem}

%% file: src/samplecomplexity.tex
\section{Sample Complexity with Value Samples}\label{sec:val}
In this section, we consider the setting in which the seller can draw i.i.d. samples from the value distribution $F$, and uses these in \Cref{alg:1} to compute a near-optimal menu and signaling scheme.



We explain the \Cref{alg:1} with the four key ingredients that we use, the last three of which are novel in this setting. First, we solve the best monotone partitional signaling scheme and menu with empirical distribution, by solving the optimization program from \citet{bergemann2022screening}. In fact, in \Cref{sec:fptas}, we will give an algorithm to do this step in polynomial time. This gives the quantile partitions of values corresponding to different signals, as in \Cref{def:1}. Second, to actually implement such a signaling scheme on the true distribution, we translate this description into the value space. There is some subtlety to this and we describe this in detail in \Cref{sec:conv}. Third, the buyer's response, i.e., the item they pick from the menu, is dictated by their signal-conditioned posterior mean, which is different on the true and the empirical distributions. Posting the empirically determined menu as is can result in a constant revenue loss. To fix this, we discount the prices in a linear-additive manner in Line 5 so that a buyer does not defect to a menu item cheaper than what was intended for them in the empirical world, since a price corresponding to a higher signal also gets a larger discount. Multiplicative discounting of the form $\tilde{p}_k = (1-\rho)p_k$ is a common tool to repair incentive-compatibility constraints (e.g., see Lemma 9 in \citet{cai2021howto}), but this would lead to a sample complexity of $\widetilde{\mathcal{O}}(1/\varepsilon^5)$ in our setting. The linear-additive discount crucially relies on the menu being bounded in size \citep{li2025sell}, which a multiplicative discount does not take advantage of. Finally, the signal-conditioned posterior means can be arbitrarily far off under $F$ and $\widehat{F}$ for signals with low marginal probability. So, in truth, we bundle all such signals, termed {\em bad}, into a \texttt{null} signal, and only attempt to repair incentive compatibility for good signals. The finitude of the menu once again means that the revenue loss due to this bundling of (individually) rare signals can not be too high.

\begin{algorithm2e}[t]
\caption{Empirical Revenue Maximization $(\mathsf{ERM})$ via Samples}
\KwIn{value distribution $F$, sample budget $n$, cutoff $\tau$, discount $\rho$.}
Construct the empirical distribution $\widehat{F}(v) = \frac{1}{n}\sum_{i=1}^n\bm{1}_{v_i\leq v}$ by drawing $n$ i.i.d. samples from $F$.\

Let $(t_k: k\in [K])$ be the quantile description of a near best monotone partitional signaling scheme for the distribution $\hat{F}$, with the direct menu $(q_k,p_k:k\in [K])$, up to $\varepsilon_{\textrm{OPT}}$ error.\

{\color{gray} Optionally, if $\varepsilon_{\textrm{OPT}}>0$, run the procedure in \Cref{sec:rep} to obtain a direct menu with non-decreasing $(p_k-c(q_k):k\in [K])$ of equal or greater revenue.}\ 

Use the quantiles $(t_k:k\in [K])$ to generate an explicit value-space description of the form $(\pi(k \mid v): \forall v\in \mathcal{V}, k\in [K])$ with respect to $\widehat{F}$, as detailed in \Cref{sec:conv}.

Designate $\bm{z}_{\tau} = \{k\in[K]: t_k-t_{k-1}<\tau\}$ as the set of {\em bad} signals.

Modify the prices as $\tilde{p}_k = p_k - k \rho$ for all $k\in [K]$.

Design a modified signaling scheme that clumps all {\em bad} signals into a new {\texttt{null} signal as
\begin{align*}
    &\tilde{\pi}(s\mid v) = \begin{cases} \pi(s \mid v) & \text{ if } s\notin\bm{z}_\tau,\\ \sum_{k\in \bm{z}_\tau}\pi(k \mid v) & \text{ if } s=\texttt{null}\end{cases}
\end{align*}

\KwOut{$\tilde{\bm{x}}_{\hat{F}}$ encoding $((q_k, \tilde{p}_k): k\in [K])$ and $(\tilde{\pi}(s\mid v):v\in \mathcal{V}, s\in [K]\cup \{\texttt{null}\})$.}}
\label{alg:1}
\end{algorithm2e}

From now on, we will use $Rev^*(F)$ to denote the optimal revenue associated with a distribution $F$, and $Rev(\bm{x}; F)$ to denote the revenue associated with a menu-signaling-scheme pair $\bm{x}$ on the same distribution. The distribution $F$ in fact plays two roles in determining the revenue (see \Cref{eq:opt1}): it appears in the objective dictating the probability with which signals are generated, and it is used to update the buyer's belief to arrive at posterior means. In more elaborate tri-variate form, for the analysis, we will let $Rev(\bm{x}; G, F)$ be the revenue when signals are generated according to $F$, yet buyers update their beliefs starting from the prior $G$. Thus, $Rev(\bm{x}; F) = Rev(\bm{x}; F, F)$.

\begin{theorem}
 \label{thm:eps3samplecomplexity}  
 For any $\varepsilon, \delta>0$, there exists a setting of $n,\tau,\rho$ and $\varepsilon_{OPT}=0$ satisfying $\tau=\widetilde{\mathcal{O}}(\varepsilon/\overline{q}K)$ and $\rho=\widetilde{\mathcal{O}}(\varepsilon/K)$ such that upon observing $n=\tilde{\mathcal{O}}\left(\frac{K^3\overline{q}^3}{\varepsilon^3}\log\frac{1}{\delta}\right)$ i.i.d.\ samples from $F$, with probability at least $1-\delta$, \Cref{alg:1} produces a menu-signaling-scheme pair $\tilde{\bm{x}}_{\hat{F}}$ satisfying
 \begin{align*}
     Rev(\tilde{\bm{x}}_{\hat{F}};F) \geq Rev^*(F) - \varepsilon.
 \end{align*}
 Further, for optimization error $\varepsilon_{OPT}>0$ in line 2, the revenue is reduced by an additive $\varepsilon_{OPT}$ term.
\end{theorem}

If we constrain our menu to a single item, we can improve the sample complexity to $1/\varepsilon^2$. In fact, \citet{bergemann2022screening} provide sufficient conditions under which such a menu may be optimal within the wider unrestricted class. We preview the key points in the analysis in \Cref{sec:valuekey}, where we also highlight the difficulty of the multi-item case, while complete proofs are given in \Cref{sec:appval}.

\begin{theorem}
 \label{thm:justone}  
 For any $\varepsilon, \delta>0$, there exists an algorithm that draws $\tilde{\mathcal{O}}\left(\frac{\overline{q}^2}{\varepsilon^2}\log\frac{1}{\delta}\right)$ i.i.d.\ samples from $F$, and with probability at least $1-\delta$, produces a signaling scheme and a menu containing a single item with revenue within $\varepsilon$ of the optimal revenue generated by  single-item menus.
\end{theorem}


\subsection{Value-space Description of Quantile-based Signals}\label{sec:conv}
Recall that the buyer chooses a menu item using the prior $F$ which is unknown to the seller. In this section, given a quantile description of a monotone partitional signaling scheme $(t_k: k\in [K])$ associated with a distribution $\widehat{F}$, we give a recipe to convert this to a value-space description: given a specific $v$, what is the corresponding distribution over signals? This matters because we are transplanting a signaling scheme designed for $\widehat{F}$ to $F$, and only a value-space description is invariant to changes in the prior distribution. For a continuous distribution $\widehat{F}$, this is straightforward as noted in \citet{bergemann2022screening}; namely, all values in the range $[\widehat{F}^{-1}(t_{k-1}), \widehat{F}^{-1}(t_k))$ are assigned the signal $k$ deterministically, where $\widehat{F}^{-1}(t) = \inf\{v:t\leq \widehat{F}(v)\}$ is the generalized inverse (or quantile) function. 
\footnote{As a convention, quantiles in $(\max_t\hat{F}^{-1}(t), 1]$ are assigned the signal corresponding to the previous quantile interval. }

The translation for discrete (or mixed) distributions is more delicate and requires randomized signaling. For all $k$ such that $\widehat{F}^{-1}(t_k)=\widehat{F}^{-1}(t_{k-1})=v_0$ (say), we choose the conditional law $\pi(k|v) = (t_k- t_{k-1})/(\widehat{F}(v_0)- \widehat{F}(v_0^-))$ for $v=v_0$ and $0$ for all other $v$, where $\widehat{F}(v_0^-)$ is the left limit of $\widehat{F}$ at $v_0$. Intuitively, only a part of the total mass $\widehat{F}(v_0)- \widehat{F}(v_0^-)$ of $v_0$ is assigned to the signal $k$. For all $k$ such that $\widehat{F}^{-1}(t_k)=v_1$ and $\widehat{F}^{-1}(t_{k-1})=v_0$ with $v_0\neq v_1$, the conditional probability of signal $k$ is

$$\pi(k|v) = \mathbf{1}\{x=v_0\} \frac{\widehat{F}(v_0)-t_{k-1}}{\widehat{F}(v_0)-\widehat{F}(v_0^-)} + \mathbf{1}\{x\in (v_0, v_1)\}+\mathbf{1}\{x=v_1\} \frac{t_{k}-\widehat{F}(v_1^-)}{\widehat{F}(v_1)-\widehat{F}(v_1^-)}.$$

The key property of this translation is that the marginal probability of generating signal $k$ under $\widehat{F}$ remains $t_k-t_{k-1}$ and $\mathbb{E}_{\widehat{F}} [v|k] = \int_{t_{k-1}}^{t_k} \widehat{F}^{-1}(x) dx/(t_k-t_{k-1})$ by an immediate application of integration by parts, exactly as intended in the quantile space.


\subsection{Key Ingredients in the Analysis}\label{sec:valuekey}
The main point of care is that for the same signaling scheme $\pi$, the posterior means conditioned on the same signal $\overline{v}_{\pi, F}(k) $ and $\overline{v}_{\pi, \widehat{F}}(k)$ can be different. This can break \texttt{IC} and \texttt{IR} constraints, resulting in a receiver of signal $k$ buying a different item than intended or nothing at all, netting $O(1)$ revenue loss. To fix this, we discount the prices. Naively, one might hope to discount the price by $q_k|\overline{v}_{\pi, F}(k) -\overline{v}_{\pi, \widehat{F}}(k)|$ for item $k$. Such conditional means can be arbitrarily far apart for rare signals; one way to see this is that the probability of signal's occurrence appears in the denominator of these quantities. This by itself might not seem bad because the contribution to the net revenue for a signal $k$ is also proportional to the probability of its occurrence. This reasoning indeed gets $\widetilde{\mathcal{O}}(1/\varepsilon^2)$ sample complexity against the benchmark of single-item menus (see \Cref{thm:justone}).

But for a multi-item menu, that is, the general case, this mode of analysis breaks down because discounting $p_k$ by itself can break incentive-compatibility for other types (recipients of other signals). So, any discount offered on $p_k$ must propagate through higher prices too, to ensure that other types do not defect to cheaper products (see \Cref{lemma:samplecomplexityimpA}). Thus, the impact of discounting the price by $\Delta$ at some signal $k$ is not localized, and affects the net revenue by $O(K\Delta)$. This is why we designate rare signals as {\em bad} signals, and only hope to repair incentive-compatibility constraints on {\em good} signals, where the conditional means can be estimated well, and hence, where limited discounting is sufficient. 

The above framework gets us a sample complexity of $\widetilde{\mathcal{O}}(1/\varepsilon^4)$, when paired with standard uniform deviation arguments (e.g., the DKW inequality \cite{Dvoretzky1956asymptotic,massart1990tight}). Such uniform arguments are needed because the boundaries of our signals are data-dependent, since we optimize over these. To improve upon this, we in fact prove a variance-aware uniform deviation inequality that says the error in conditional means scales as $1/\sqrt{\nu n}$ (upto lower order terms), instead of $1/\nu\sqrt{n}$, where $\nu$ is the frequency of the signal's occurrence (see \Cref{lem:conc}). Our results imply a variance-aware DKW inequality whose existence has been posed as an open question in \citet{guo2019settling}, but in fact we show that such results can be derived with relative ease from known relative-deviation-style uniform convergence bounds for VC classes \citep{haussler1992decision,li2001improved,bartlett2005local,Maurer2009EmpiricalBB}.

%% file: src/dynamic_programming.tex
\section{Additive FPTAS for Revenue Maximization}\label{sec:fptas}
\begin{algorithm2e}[t]
\caption{Dynamic Programming (\texttt{DP}) for Optimal Menu and Signaling Scheme}
\label{alg:dp}
\KwIn{CDF $F$ on $[0,1]$; quantile grid $\bm{t}=(t_0,t_1,\ldots,t_{n_t})$ with $t_0=0,t_i=i\varepsilon$; quality grid $\mathcal{E}_q=\{q_0,q_1,\ldots
,q_{n_q}\}$ with $q_0=0$, $q_j=\qmin+j\varepsilon$; max. menu size $K$ and $A(t)$}
$\texttt{Rev}[0,0,0] := 0$ and
$\texttt{Rev}[k,r,j] := -\infty$ for all $(k,r,j)\neq (0,0,0)$\tcp{Initialization}

\For{$k = 1,\ldots,K$;\ $r = 1,\ldots,m$;\ $j = 1,\ldots,n_q$}{
$\texttt{Rev}[k,r,j] := \max_{0\leq \ell<r,\;0\leq j'\leq j}\bigl\{\texttt{Rev}[k-1,\ell,j']+\mathcal{U}[\ell,r,j',j]\bigr\}$\;
$\texttt{ptr}[k,r,j] := \argmax_{0\leq \ell<r,\;0\leq j'\leq j}\bigl\{\texttt{Rev}[k-1,\ell,j']+\mathcal{U}[\ell,r,j',j]\bigr\}$\;
}
$(k^\star,j^\star) := \argmax_{1\leq k\leq K,\;1\leq j\leq n_q}\texttt{Rev}[k,m,j], r_{k^\star}:=n_t, j_{k^\star}:=j^\star, r_0:=0$\;
\For{$k = k^\star-1$ \KwTo $1$}{
    $(r_k,\,j_k) := \texttt{ptr}[k+1,\,r_{k+1},\,j_{k+1}]$\;
}
\For{$k = 1$ \KwTo $k^\star$}{
$p_k := \mu_F(r_{k-1},r_k)\,q_{j_k} - \sum_{\ell=1}^{k-1}\bigl(\mu_F(r_\ell,r_{\ell+1})-\mu_F(r_{\ell-1},r_\ell)\bigr)\,q_{j_\ell}$\;
}

\KwOut{Menu $\{(p_k,q_{j_k})\}_{k=1}^{k^\star}$ and signaling scheme $\mathcal{I}=(k^\star,(t_{r_k}:k\in [k^\star]))$}

\end{algorithm2e}
As usual for an FPTAS, the key is to formulate a dynamic program on suitably small state space. What makes this possible is the fact that while scanning values from left to right, or quantiles from bottom to top, with the intention of pooling these, it is enough to remember just the quality assigned to the last pool (instead of all qualities). This relies on the structure of incentive compatibility constraints in the setting of \citet{MUSSA1978301}, as we detail next. The theorem as stated, and proved in \Cref{sec:appfptas}, only holds for discrete distributions. But the generalization to arbitrary distributions is {\em free}, given sampling access, because the empirical distribution is discrete, and we can plug this algorithm into line 2 of \Cref{alg:1}.

\begin{theorem}[{Additive-FPTAS}] For any $\varepsilon_{\mathsf{OPT}}>0$, given a discrete distribution with support size $N$, \Cref{alg:dp} runs in $\mathsf{poly}(K,N,\overline{q},1/\varepsilon_{\mathsf{OPT}})$ time and outputs a menu $\{(p_k, q_{j_k})\}_{k=1}^{k^\star}$  and
    a monotone paritional signaling scheme $\mathcal{I}=(k^\star,(t_{r_k}:k\in [k^\star]))$ with revenue exceeding ${Rev}^\star(F)-\varepsilon_{\mathsf{OPT}}$ by setting the resolution of grid as $\varepsilon = (\varepsilon_{\mathsf{OPT}})^2/(K(6\qmax +1+\qmax))^2$.
\end{theorem}

We start by simplifying the optimization problem (\Cref{eq:opt1}). It is well known that the participation constraint for the lowest type binds and incentive compatibility constraints for all other types bind with the immediately lower type \citep{maskin1984monopoly}. Moreover, the constraints imply monotonicity of the qualities with respect to the types (or values). That is $\bar{v}_i\cdot q_i-p_i=\bar{v}_i\cdot q_{i-1}-p_{i-1}$ with $q_{0}=p_0=0$ and $q_i\geq q_{i-1}\ \forall i$. This gives us the so called payment formula $p_i = \bar{v}_iq_i -\sum_{j=1}^{i-1}(\bar{v}_{j+1}-\bar{v}_{j})q_j$. Thus, we can write the optimization problem in the quantile space, which we use in the dynamic programming formulation, as
\begin{align*}
   &\texttt{L}^s_{F}:\   \max_{\left((I,\bm{t}),(q_i)_{i\in [I]})\right)}\sum_{i=1}^I \bar{v}_i(1-t_{i-1})(q_{i}-q_{i-1})- (t_{i}-t_{i-1})c(q_i)\\
     &q_i
     \geq q_{i-1} \ \forall i\in\{1,\ldots,I\}\ \text{with $q_0=0$},\quad t_0=1-t_{I}=0,\ \text{and}\ \bar{v}_i=\tfrac{\int_{t_{i-1}}^{t_{i}}F^{-1}(t)\mathrm{d}t}{t_i-t_{i-1}}
\end{align*}
 Let $\mathcal{E}_q=\{q_0, q_1,\ldots,q_{n_q}\}$ with $q_0=0$ be the quality grid, where $q_0$ denotes the outside option and $q_j:=\qmin+j\varepsilon$ for $j\in\{1,\ldots,n_q\}$, with $n_q=\lceil \frac{\qmax-\qmin}{\varepsilon}\rceil$, and $\mathcal{E}_t=\{t_0,\ldots,t_{n_t}\}$ with $t_0=1-t_{n_t}=0$ be the quantile grid with $t_i=i\varepsilon$ and $n_t=1/\varepsilon\in\mathbb{Z}_{>0}$. Define $A(t) = \int_{0}^tF^{-1}(s)\,\mathrm{d}s\ \forall t\in[0,1]$. For indices $0\leq \ell<r\leq m$, define $\mu_F(\ell,r) := \frac{A(t_r)-A(t_\ell)}{t_r-t_\ell}$
the posterior mean of the quantile interval $(t_\ell,t_r]$ under $F$.  For the dynamic programming we store $\texttt{Rev}[k,r,j]$ which denotes the maximum revenue from partitioning $(0,t_r]$ into $k$ nonempty quantile grid intervals with the last interval assigned quality $q_j\in\mathcal{E}_q$. Moreover, we store $\mathcal{U}[\ell,r,j',j]$ for quality indices $0\leq j'\leq j\leq n_q$ which we define as the \emph{transition utility}
\begin{align*}
    \mathcal{U}[\ell,r,j',j] := \mu_F(\ell,r)\,(1-t_\ell)\,(q_j-q_{j'}) -  c(q_j)(t_r-t_\ell),
\end{align*}
i.e. the revenue contribution from assigning quality $q_j$ to the interval $(t_\ell,t_r]$, given that the preceding interval was assigned quality $q_{j'}$. The dynamic programming recursion immediately follows where one needs to maintain the monotonicity of quality being assigned.



%% file: src/query.tex
\section{Sample Complexity with Demand Queries}
\begin{algorithm2e}[!t]
\caption{Empirical Revenue Maximization via Demand Queries}
\KwIn{query budget $n$, cutoff $\tau$, discount $\rho$.}
Fix a menu $\mathcal{M}=\{(q, p(q):={(q-\underline{q})^2}/2(\overline{q}-\underline{q})): q\in [\underline{q}, \overline{q}]\}$.\

\For{$j=1,\dots n$}{
Choose a signaling scheme with two signals $\pi_j(0|v)=1-\pi_j(1|v)=v^{j-1}$.\

Observe the generated signal $s_j$ and the purchased quality $q_j$ when the menu $\mathcal{M}$ was coupled with a signaling scheme $\pi_j$.

Let $q'_j = (q_j-\underline{q})/(\overline{q}-\underline{q})$ be the normalized quality.

With the convention that $m_{0}=1$, set $m_j=\begin{cases}
    q'_jm_{j-1} & \text{ if } s_j=0,\\
    m_1 - q'_j(1-m_{j-1}) & \text{ if } s_j=1.
\end{cases}$
}

\If{the value distribution is discrete}{
For a discrete value distribution with support $\{x_1, \dots x_{n+1}\}$ of size $n+1$, calculate the pmf $p\in \Delta([n])$ by solving the linear system $Vp=\bm{m}$, where $\bm{m}=[m_0, m_1, \dots m_n]^\top$ and $V\in \mathbb{R}^{(n+1)\times (n+1)}$ is a Vandermonde matrix with $V_{i,j} =x_j^i$.\

Let $\widehat{F}$ be the corresponding CDF.
}
\Else{
Let $L_j(x) =\sum_{i=0}^j a_{ji}x^i$ be the expansion of the $j^{th}$ Legendre polynomial shifted to $[0,1]$.\

Calculate $\alpha_j = \sum_{i=0}^j a_{ji} (1-m_{i+1}/(i+1))$ and $\tilde{F}_j(x) = \sum_{i=0}^j \alpha_i L_i(x)$, which is the degree-$j$ Legendre projection of the CDF.

Compute the de la Vallee Poussin mean $\hat{F}_{\lceil n/2\rceil}(x) = \frac{1}{\lceil n/2\rceil} \sum_{i=\lceil n/2\rceil}^{2\lceil n/2\rceil-1} \tilde{F}_i(x)$.\

Let $\hat{F}$ be the smallest non-negative non-decreasing upper envelope of $\hat{F}_{\lceil n/2\rceil}$.
}

Let $(t_k: k\in [K])$ be the quantile description of the best monotone partitional signaling scheme for the distribution $\hat{F}$, with the direct menu $(q_k,p_k:k\in [K])$.\

Use the quantiles $(t_k:k\in [K])$ to generate an explicit value-space description of the form $(\pi(k \mid v): \forall v\in \mathcal{V}, k\in [K])$ with respect to $\widehat{F}$, as detailed in \Cref{sec:conv}.

Designate $\bm{z}_{\tau} = \{i\in[I]: t_i-t_{i-1}<\tau\}$ as the set of {\em bad} signals.

Modify the prices as $\tilde{p}_k = p_k - k \rho$ for all $k\in [K]$.

Design a modified signaling scheme that clumps all {\em bad} signals into a new {\texttt{null} signal as
\begin{align*}
    &\tilde{\pi}(s\mid v) = \begin{cases} \pi_k(s \mid v) & \text{ if } s\notin\bm{z}_\tau,\\ \sum_{k\in \bm{z}_\tau}\pi_k(k \mid v) & \text{ if } s=\texttt{null}\end{cases}
\end{align*}

\KwOut{$\tilde{\bm{x}}_{\hat{F}}$ encoding $((q_k, \tilde{p}_k): k\in [K])$ and $(\tilde{\pi}(s\mid v):v\in \mathcal{V}, s\in [K]\cup \{\texttt{null}\})$.}}
\label{alg:2}
\end{algorithm2e}
We consider a setting where the seller deploys signaling schemes and menus of their choice, and observes how a buyer with a randomly drawn value responds to it. Despite the stochasticity in values, we will see that it is possible to extract the moments of the value distribution exactly in this setting, without any error. We couple this with results from approximation theory \citep{trefethen2019approximation} to estimate the underlying value distribution. We define the setting first and then give the algorithm.

\begin{definition}
    A demand query accepts as input a signaling scheme $(\mathcal{S},\{\pi(\cdot|v)\}_{v\in\mathcal{V}})$ and a menu $((p(z),q(z))_{z\in Z })$ for some index set $Z$ and outputs the signal $s\sim \pi(\cdot |v)$ made available to the buyer along with the item $(q,p)$ selected by the buyer, where the value of the buyer $v$ is drawn randomly from $F$. 
\end{definition}

\subsection{Approximating Discrete Distributions}

The main observation is that in this model we can exactly simulate bounded linear functionals of the pmf (or the pdf) corresponding to the common prior $F$. We specifically chose moments. Interestingly, this makes our queries non-adaptive. 

\begin{theorem}\label{thm:disc}
For any discrete distribution $F$ with support size $n+1$, when \Cref{alg:2} is run for $n$ queries, with $\tau=\rho=0$, it returns a revenue-maximizing menu and signaling scheme pair $\tilde{\bm{x}}_{\widehat{F}}$ that attains the optimal revenue, that is, $Rev(\tilde{\bm{x}}_{\widehat{F}}) ={Rev}^*(F)$. 
\end{theorem}

We first begin with the following moment recovery lemma.

\begin{lemma}
    For any $j\leq n$, $m_j=\mathbb{E}_F[v^{j}]$ as recovered on line 6 of \Cref{alg:2}.
    \label{lemma:momentrecovery}
\end{lemma}   
\begin{proof}
Upon observing any signal $s$, the quality of the item chosen by the buyer is $\argmax_{q\in [\underline {q}, \overline{q}]}q\overline{v}_{\pi, F}(s)-p(q)$. First-order conditions dictate that $q'_j=\overline{v}_{\pi, F}(s)$, thus we recover the posterior mean exactly. Now, in any round $n$, we have $\overline{v}_{\pi, F}(0) = \tfrac{\mathbb{E}_F[v^k]}{\mathbb{E}_F[v^{k-1}]} \text{ and }\overline{v}_{\pi, F}(1) = \tfrac{\mathbb{E}_F[v]-\mathbb{E}_F[v^k]}{1-\mathbb{E}_F[v^{k-1}]}$.
Thus, inductively, we have that $m_j = \mathbb{E}_F[v^{j}]$.
\end{proof}

Given a sufficient number of moments, Vandermonde inversion guarantees exact pmf recovery. 

\begin{proof}[Proof of \Cref{thm:disc}]
Using Lemma~\ref{lemma:momentrecovery} we can see that $\mathbb{E}_F[v^j]=\sum_{i=1}^{n+1} p_i x_i^j = V_j^{\top} p$, where $V_j$ is the $j^{th}$ row of the Vandermonde matrix defined on line 8 of \Cref{alg:2}. Hence, the true pmf satisfies $Vp=\bm{m}$. Because the Vandermonde matrix is invertible as long as $x_i\neq x_j$ for all $i<j$, the pmf reconstructed by the algorithm is exact, and no revenue is lost.
\end{proof}

\subsection{Approximating the CDF for Continuous Distributions}

The continuous case is trickier. We would like to approximate the CDF $F$ via orthogonal polynomial families in the $L_\infty$ sense, that is, in the Kolmogorov metric. We instead settle on the best $L_2$ approximation, since those can be computed via inner products on the Lebesgue measure; the latter restriction rules out well-loved classes like the Chebyshev family \citep{cheney2009course}. However, Legendre projections are known to have suboptimal $L_\infty$ approximation guarantees, so we use de la Vallee Poussin mean of Legendre projections which in effect stabilizes the latter, and uses twice the number of moments to give near-optimal guarantees matching the best $L_\infty$ approximation à la Jackon's theorem \citep{mastroianni2008vallee}. A final hurdle is that moments of a distribution naturally correspond to inner products of the moment curve with the pdf, and we would like to compute inner products with the CDF. This gap can be bridged by an application of integration by parts. We supply the details in the Appendix~\ref{sec:appdem}.
$$
 \textstyle\int_0^1 x^j F(x) dx = 1 - \tfrac{1}{j+1}\int_0^1 x^{j+1} f(x) dx.
$$
\begin{theorem}\label{thm:notdisc}
For any $r\in \mathbb{Z}_{>0}$ and  any continuous distribution $F$ where the $r^{th}$ derivative of the CDF is bounded as $B^r$, for some constant $B>0$, when \Cref{alg:2} is run for $n=\mathcal{O}(\frac{\overline{q}^{2/r}K^{2/r}B}{\varepsilon^{2/r}})$ queries and suitable $\tau, \rho$, it produces a menu-signaling-scheme pair $\tilde{\bm{x}}_{\hat{F}}$ satisfying
 \begin{align*}
    Rev(\tilde{\bm{x}}_{\hat{F}};F) \geq Rev^*(F) - \varepsilon.
 \end{align*}
\end{theorem}
As an immediate corollary, for bounded and Lipschitz pdfs we need $1/\varepsilon^2$ and $1/\varepsilon$ queries, respectively. For analytic functions, we choose $r=\log (1/\varepsilon)$ to get polylogarithmic in $1/\varepsilon$ queries.

%% file: src/robustregret.tex
\section{Regret under Joint Learning}
Followed by  \cite{zu2025learning}, we consider a more general setting where the seller interacts with a buyer on each day, commits to a menu and signaling scheme, and observes the purchase decision and the value of the buyer at the end of the day. More importantly, we allow the case when the population of buyers are learning, perhaps exogenously through reviews of the item on the platform. This setting thus subsumes the previous settings, however, the analysis closely mirrors the earlier one once we introduce the following relaxed behavioral assumption. Since the buyers are learning, the seller asks, what does the belief the buyer who arrives on that day hold? If the seller knew the belief of the buyer, he could solve the offline problem. Since he does not know the buyer's belief precisely, he robustifies the menu and signaling scheme which performs well for a set of beliefs, which the seller iteratively updates knowing that buyers are learning. This idea is captured in the following definition, where the seller confidently assumes that the arriving buyer's belief on day $t$ lies in the set $\mathcal{C}(\hat{F}_t)$ which can simply be $\{F^\star\}$, the \textit{unknown} common belief, $\mathcal{B}_0$ the sellers prior belief about buyers' belief or simply $\hat{F}_t$, the empirical distribution on day $t$.
\begin{definition}
We say that an algorithm  which executes $\tilde{\bm{x}}_t$ in each round is $(\beta,\bm{\kappa})$-robustly revenue efficient, if there exists \textit{data-dependent} sets $\mathcal{C}_t$ such that we have $(i)$ coverage i.e. $\text{Pr}\left(\cap_{t=1}^T\mathcal{C}_t\ni F^\star\right)\geq 1-\beta$ and $(ii)$ robustness, i.e. for $\mathcal{H}_t$ the set of all possible histories up to time $t$, we need for all $t\in[T]$ and for all $h_t\in \mathcal{H}_t$, conditioned on the history $h_t$ that  $\sup_{G\in \mathcal{C}(\hat{F}_t)}\max_{\bm{x}}\texttt{Rev}(\bm{x};\hat{F}_t,\hat{F}_t)-\texttt{Rev}(\bm{\tilde{x}}_t;G,\hat{F}_t)\leq \kappa_t$.
\label{defn:betakappa}
\end{definition}
\begin{assumption} For a $(\beta,\bm{\kappa})$-robustly revenue efficient algorithm w.r.t. sets $\mathcal{C}_t$, the buyers' prior $F^b_t$ at round $t$ is such that $\text{Pr}(\mathcal{C}_t\ni F^b_t)=1\ \forall\ t\in[T]$. The seller knows the set $\mathcal{B}_0$ s.t. $\text{Pr}(\mathcal{B}_0\ni F^\star)=  1$. 
\label{asm:behavioral}
\end{assumption}

Although the Definition~\ref{defn:betakappa} allows for arbitrary sets $\mathcal{C}$, one is interested in sets which are statistically achievable by the buyer. We propose Algorithm~\ref{algo:online},  where the seller chooses parameter sequence based on the assumed sets $\mathcal{C}$. Throughout the following results we take the sets $\mathcal{C}_t(\varepsilon_t):=\{G\in \Delta
    (\mathcal{V}):\sup_v|G(v)-\hat{F}_t(v)|\leq \varepsilon_t
    \}$ for all $t\in[T]$, which are $\varepsilon_t>0$ sized balls centered at $\hat{F}_t$ in the Kolmogorov metric. First, we show the space parameters $\beta,\kappa$ for which the algorithm is robustly revenue efficient.

\begin{algorithm2e}[t]
\caption{{Online Empirical Menu \& Information Design}}
\KwIn{$\tau_t, \varepsilon_t, h(\tau_t, \varepsilon_t)$ for $t=1,\ldots,T$}

Initialize $\hat{F}_1 \in \mathcal{B}_0$\;

\For{$t = 1, \ldots, T$}{
    Call \Cref{alg:1} with the distribution $\hat{F}_t$, cutoff $\tau_t$ and discount $\rho_t$ to obtain $\tilde{\bm{x}}_{\hat{F}_t}$\;
    Execute $\tilde{\bm{x}}_{\hat{F}_t}$ and observe $v_t$ and update $\hat{F}_{t+1}(v) = \frac{1}{t}\sum_{s=1}^{t}\mathbf{1}_{\{v_s \leq v\}}$\;
}
\label{algo:online}
\end{algorithm2e}

\begin{theorem}
 The Algorithm~\ref{algo:online}
    with any cutoff sequence $\{\tau_t\}_{t=1}^T$ and for any sequence $\{\varepsilon_t\}_{t=1}^T$ with $\varepsilon_t\in (0,1]$ and discount sequence given by $\rho_t = 4\Bar{q}\varepsilon_t/\tau_t$
    is $(\beta,\bm{\kappa})$-robustly revenue efficient for sets $\mathcal{C}_t(\varepsilon_t)$ for all $\beta\geq \sum_{t\in[T]}2\exp\left(-t\varepsilon_t^2/2\right)$ and $\kappa_t = K\Bar{q}(4\varepsilon_t/\tau_t + \tau_t)$.
\end{theorem}

We measure the performance of the algorithm through cumulative regret defined as  $\texttt{Reg}:=\sum_{t=1}^T\texttt{Rev}^\star(F^\star)-\texttt{rev}_t(\bm{x}_t;F^b_t)$ where $\texttt{rev}_t(\bm{x}_t;F^b_t)$ denotes the revenue generated by seller on day $t$ upon committing to $\bm{x}_t$ when the arriving buyer hold belief $F^b_t$.

\begin{theorem}  Algorithm~\ref{algo:online}
    with cutoff $\tau_t=\sqrt{\varepsilon_t}$ and  discount $\rho_t = 4\Bar{q}\sqrt{\varepsilon_t}$, for the sets $\mathcal{C}_t(\varepsilon_t)$ with $\varepsilon_t=\sqrt{3\log T/ t}$ gives the cumulative regret $\texttt{Reg}
    \leq    \mathcal{O}\left( K\bar{q}(\log T)^{1/4}T^{3/4}\right)$ with probability at least $1-3T^{-1/2}$.
\end{theorem}

In \Cref{sec:appreg}, we prove the above results and also show that there are (more intricate) statistically plausible sets $\mathcal{C}_t$ for which we obtain a $\tilde{\mathcal{O}}(T^{2/3})$ regret upper bound.

%% file: src/conclusion.tex
\section{Conclusion}
We study a novel learning problem motivated by the fact that sellers and algorithmic platforms today have access to vast user trails that they can use to forecast the personalized value of niche products for a specific user. We give learning algorithms that reflect how a seller might use past experience to craft the best menu of products and prices, along with a recommendation for individual customers on what products to buy. With access to samples of values, we achieved a sample complexity of ordert $1/\varepsilon^3$. We then showed that this can be improved significantly in a model where the seller can observe the buyers' behaviors to carefully chosen menu and recommendations, under some circumstances making the sample complexity near-constant. Finally, we give an efficient algorithm to compute a near-optimal menu and signaling scheme, despite the underlying problem being non-convex. Our work brings a learning-theoretic perspective to the interface between information and (traditional) mechanism design.

\paragraph{Acknowledgements} TP is supported by the Balas Fellowship Award and Ph.D. funding from the Tepper School of Business. We thank Rattana Pukdee, Kiriaki Fragkia and Andreas Kalavas for helpful comments.

%% file: src/appendix.tex
\section*{Appendix}
In \Cref{sec:appval}, we complete the proof for sample complexity with value samples. \Cref{sec:appdem} supplies the proofs for sample complexity of demand queries. In \Cref{sec:appreg} and \Cref{sec:appfptas}, we complete the proofs for the joint-learning-based regret setting and prove the correctness of the FPTAS, respectively. \Cref{app:equiv} provides formal proofs for the signaling transformation discussed in \Cref{sec:conv}.

\section{Sample Complexity with Value Samples}\label{sec:appval}

\subsection{Proof of Theorem~\ref{thm:eps3samplecomplexity}}
We decompose the error as 
    \begin{align*}
          Rev^*(F) -   Rev(\tilde{x}_{\hat{F}};F) & =  \underbrace{Rev^*(F) - Rev^*(\hat{F})}_{\texttt{Term1}}  + \underbrace{Rev^*(\hat{F})- Rev(\tilde{x}_{\hat{F}};F,F)}_{\texttt{Term2}}.
    \end{align*}

Throughout the analysis, we condition on the event in \Cref{lem:conc}, which holds with high probability.
\begin{lemma}\label{lem:conc}
There exists a universal constant $C$ such that for any $\delta>0, n>1$ and distribution $F$, with probability $1-\delta$, for all monotone paritional signaling scheme $\pi$ with $K$ signals, we have that
\begin{align} \label{eq:condmean}
    &\left|\Pr_{\pi\circ F}({k}) - \Pr_{\pi\circ \widehat{F}}({k}) \right| \leq C\left(\sqrt{\min\{\Pr_{\pi\circ {F}}(k),\Pr_{\pi\circ {\widehat{F}}}({k})\}\frac{\log (n /\delta)}{n}} + \frac{(\log n /\delta)}{n}\right), &\\ 
     &\left|\mathbb{E}_{F}[v|k] - \mathbb{E}_{\widehat{F}} [v|k] \right| \leq C\left(\sqrt{\frac{\log (n /\delta)}{n\max\{\Pr_{\pi\circ {F}}(k),\Pr_{\pi\circ {\widehat{F}}}({k})\}}} + \frac{\log (n /\delta)}{n\max\{\Pr_{\pi\circ {F}}(k),\Pr_{\pi\circ {\widehat{F}}}({k})\}}\right),\nonumber 
\end{align}
for all $k\in [K]$, where $\widehat{F}$ is a $n$-point empirical distribution sampled from $F$.

\end{lemma}

First, we upper bound \texttt{Term 1}. Let $\bm{x}_F$ be the optimal pair of menu $((q_k, p_k): k\in [K])$ and monotone partitional signaling scheme $(t_k: k\in [K])$ under $F$. Let $\bm{z}_{\tau}=\{k\in [K]: t_k-t_{k-1}< \tau\}$. Define $\tilde{\bm{x}}_F$ with signaling scheme $\tilde{\pi}$ modified so that all signals in $\bm{z}_\tau$ are mapped to a \texttt{null} signal:
\begin{align*}
    \tilde{\pi}(s\mid v) = \begin{cases} \pi(s\mid v) & \text{if } s\notin\bm{z}_\tau,\\ \sum_{k\in\bm{z}_\tau}\pi(k\mid v) & \text{if } s=\texttt{null},\end{cases}  
\end{align*}
and the prices discounted as $\tilde{p}_k = p_k-k\rho$. Then we have
    \begin{align*}
        \texttt{Term1}&:= Rev^*(F) - {Rev}(\tilde{x}_F;\hat{F},\hat{F}) + {Rev}(\tilde{x}_F;\hat{F},\hat{F}) - Rev^*(\hat{F}) \\
        &\leq Rev^*(F) - {Rev}(\tilde{x}_F;\hat{F},\hat{F})
    \end{align*}
    where we appeal to the fact that $Rev^*(\hat{F})$ is the optimal revenue for the empirical distribution $\widehat{F}$. Now, we use \Cref{lemma:samplecomplexityimpA}, which places an upper bound of $\tilde{\mathcal{O}}(\overline{q}K/n^{1/3})$ on \texttt{Term 1}, by choosing $\rho=\widetilde{{\Theta}}(\overline{q}/n^{1/3})$ and $\tau=\widetilde{{\Theta}}(1/n^{1/3})$.

\begin{lemma}\label{lemma:samplecomplexityimpA}
    Consider any two value distributions $F,\widehat{F}$ such that for all monotone partitional signaling schemes $\pi$ with $K$ signals, we have for all $k\in [K]$ that
    \begin{align*}
    \left|\Pr_{\pi\circ F}({k}) - \Pr_{\pi\circ \widehat{F}}({k}) \right| &\leq \varepsilon_0,  \text{ and }  
     \left|\mathbb{E}_{F}[v|k] - \mathbb{E}_{\widehat{F}} [v|k] \right| \leq \frac{\varepsilon_1}{{\Pr_{\pi\circ F}({k})}} + \frac{\varepsilon_2}{\sqrt{\Pr_{\pi\circ F}({k})}}. 
\end{align*}
    Let $\bm{x}_F$ be $\varepsilon_{\textrm{OPT}}$-suboptimal solution for distribution $F$ comprising of an incentive-compatible menu and monotone partitional signaling scheme, for which $p_k-c(q_k)$ is in non-decreasing in $k$. Further, let $\tilde{\bm{x}}_F$ be a modification of $\bm{x}_F$ with all signals with marginal probability smaller than $\tau$ (with respect to $F$) mapped to a {\em null} signal and price modifications $\tilde{p}_k = p_k - k\rho$, where $\rho=2\overline{q}(\varepsilon_1/\tau+\varepsilon_2/\sqrt{\tau})$, $p_k$ is the corresponding price for signal $k$ in $\bm{x}_F$. Then, we have $$
   {Rev}(\bm{x}_F;F,F) - {Rev}(\tilde{\bm{x}}_{F},\widehat{F},\widehat{F})\leq \overline{q}K(\tau + 2(\varepsilon_1/\tau+\varepsilon_2/\sqrt{\tau})+\varepsilon_0) .$$
\end{lemma}

Similarly, for $\texttt{Term2}$, let $\bm{x}_{\widehat{F}}$ be an $\varepsilon_{OPT}$-suboptimal pair of menu $((q_k, p_k): k\in [K])$ and monotone partitional signaling scheme $(t_k: k\in [K])$ this time under $\widehat{F}$. Let $\bm{\tilde{z}}_{\tau}=\{k\in [K]: t_k-t_{k-1}< \tau\}$. Let $\bm{\tilde{x}}_{\widehat{F}}$ be the analogous modification where all signals in $\bm{\tilde{z}}_{\tau}$ are mapped to a null signal, and prices are linear-additively discounted by $\rho$ at all indices. Now, applying \Cref{lemma:samplecomplexityimpA}, with the roles of $F$ and $\widehat{F}$ reversed, we get 
    \begin{align*}
        \texttt{Term2}&:= Rev^*({\widehat{F}}) - {Rev^*}(\bm{x}_{\widehat{F}};{\widehat{F}},{\widehat{F}}) + {Rev}(\bm{x}_{\widehat{F}};{\widehat{F}},{\widehat{F}}) - Rev^*(\bm{\tilde{\bm{x}}}_{\widehat{F}}; F,F) \\
        &\leq \varepsilon_{\textrm{OPT}}+Rev(\bm{x}_{\widehat{F}};{\widehat{F}},{\widehat{F}}) - {Rev}(\tilde{\bm{x}}_F;\hat{F},\hat{F}) \\
        &\leq \varepsilon_{\textrm{OPT}} + \tilde{\mathcal{O}}(\overline{q}K/n^{1/3}).
    \end{align*}

\subsection{Proof of Lemma~\ref{lemma:samplecomplexityimpA}}
\label{proof:lemmaimpA}
Since $\bm{x}_F$ encodes an incentive-compatible menu $((q_k, p_k):k\in [K])$ under the distribution $F$, we know that for all $k\in [K]$ that \begin{align*}
q_k\cdot \mathbb{E}_{\hat{F}}[v|k] -p_k &\geq q_k\cdot \mathbb{E}_{F}[v|k] - p_k - \overline{q}\left(\frac{\varepsilon_1}{{\Pr_{\pi\circ F}({k})}} + \frac{\varepsilon_2}{\sqrt{\Pr_{\pi\circ F}(k)}}\right) \\
&\geq \max\{0, \max_j \{q_j\cdot \mathbb{E}_{F}[v|k] - p_j\}\} - \overline{q}\left(\frac{\varepsilon_1}{{\Pr_{\pi\circ F}({k})}} + \frac{\varepsilon_2}{\sqrt{\Pr_{\pi\circ F}(k)}}\right)\\
&\geq \max\{0, \max_j \{q_j\cdot \mathbb{E}_{\widehat{F}}[v|k] - p_j\}\} - 2\overline{q}\left(\frac{\varepsilon_1}{{\Pr_{\pi\circ F}({k})}} + \frac{\varepsilon_2}{\sqrt{\Pr_{\pi\circ F}(k)}}\right)
\end{align*}
Now by implementing the modified prices $\tilde{p}_k = p_k - k\rho$, it is clear that participation constraints are upheld for $\widehat{F}$ for all (good) signals in not $\bm{z}_\tau$. We do not claim to uphold incentive compatibility under $\widehat{F}$, only that any buyer receiving signal $k$ does not purchase an item indexed by a lower signal, as long as $k\not\in \bm{z}_\tau$. To observe this, note for all $j<k$ that \begin{align*}
    q_k\cdot \mathbb{E}_{\hat{F}}[v|k] -\widetilde{p}_k &= q_k\cdot \mathbb{E}_{\hat{F}}[v|k] -p_k + k\rho\\
    &\geq q_j\cdot \mathbb{E}_{\hat{F}}[v|k] -p_j + k\rho -\rho\\
    &= q_j\cdot \mathbb{E}_{\hat{F}}[v|k] -\tilde{p}_j + (k-j)\rho -\rho\\
    &\geq q_j\cdot \mathbb{E}_{\hat{F}}[v|k] -\tilde{p}_j.
\end{align*}
Thus, each receipt of a good signal $k\not\in \bm{z}_\tau$ agrees to pay at least $p_k-K\rho$ and generates a margin at least $p_k-c(q_k)-K\rho$ for the buyer under the distribution $\widehat{F}$. Thus, we have $${Rev}(\bm{x}_F;F,F) - {Rev}(\tilde{\bm{x}}_{F},\widehat{F},F) \leq K(\rho + \overline{q}\tau),$$ 
even assuming that in the worst case those assigned the \texttt{null} signal refuse to purchase anything. Finally, using the observations that the marginal probability of signals is similar under $F$ and $\widehat{F}$:
$${Rev}(\tilde{\bm{x}}_{F};\widehat{F},F) - {Rev}(\tilde{\bm{x}}_{F},\widehat{F},\widehat{F}) \leq \overline{q}K\varepsilon_0.$$
Combining the last two inequalities concludes the proof.

\subsection{Proof of Lemma~\ref{lem:conc}}
\label{appsec:conc}
Our plan is to appeal to variance-sensitive versions of uniform convergence bounds for function classes with bounded pseudo-dimension. We start with a result of \citet{Maurer2009EmpiricalBB}, which is stated in terms of $L_\infty$ covering number on $n$-sample distributions, a term we will quickly rid ourselves of. (See \citet{van2023weak} for the definition.)

\begin{theorem}[\cite{Maurer2009EmpiricalBB}]
    Let $X$ be a random variable with values in a set $\mathcal{X}$ with distribution $F$ and let $\mathcal{F}$ be a class of hypotheses $f : \mathcal{X} \rightarrow [0, 1]$. Fix $\delta \in (0, 1)$, $n \ge 16$ and set
$$
\mathcal{M}(n) = 10\mathcal{N}_\infty(1/n, \mathcal{F}, 2n).
$$
Then with probability at least $1 - \delta$ in the random vector $\mathbf{X} = (X_1, \dots, X_n) \sim \mu^n$ we have
$$
\mathbb{E}_{x\sim F}[f(x)] - \mathbb{E}_{x\sim \widehat{F}}[f(x)] \leq \sqrt{\frac{18V_n(f, \mathbf{X}) \ln(\mathcal{M}(n)/\delta)}{n}} + \frac{15 \ln(\mathcal{M}(n)/\delta)}{n-1}, \forall f \in \mathcal{F},
$$
where $V_n(f,\mathbf{X}) = \sum_{i<j}(f(x_i)-f(x_j))^2/n(n-1)$ is the sample variance, and $\widehat{F}$ is the empirical distribution on the dataset $\mathbf{X}$.
\end{theorem}

For real-valued function classes with bounded pseudo-dimension, the above statement can be stated in terms of pseudo-dimension.

\begin{definition}
   A function class $\mathcal{F}:\mathcal{X}\to \mathbb{R}$ has pseudo-dimension $d$ if $d$ is the maximum number for which there is $d$-sized set $\mathcal{S}=\{x_1,\ldots,x_d\}\subseteq \mathcal{X}$ and real numbers $r_1,\ldots,r_d$ such that for each $b\in \{0,1\}^d$ there is a function $f_b\in \mathcal{F}$ with $\text{sign}(f_b(x_i)-r_i)=b_i\ \forall i\in[d]$.
\end{definition}

\begin{theorem}[\cite{Anthony_Bartlett_1999}]
   Let $\mathcal{F}$ be a set of real functions from a domain $\mathcal{X}$ to the bounded interval $[0, B]$. Let $\epsilon > 0$ and suppose that the pseudo-dimension of $\mathcal{F}$ is $d$. Then
\begin{align*}
    \mathcal{N}_\infty (\epsilon, F, m) \leq \sum_{i=1}^{d} \binom{m}{i} \left(\frac{B}{\epsilon}\right)^i,
\end{align*}
which is less than $(emB/(\epsilon d))^d$ for $m \ge d$.
\end{theorem}

Now, we perform three steps at once: one, we compose the last two results; two, we apply the results to the class $\mathcal{F}'=\{1-f:f\in \mathcal{F}\}$ in addition to the original class; three, we note for any $f:\mathcal{X}\to[0,1]$ that $V_n(f, \mathbf{X}) = \frac{n}{n-1}(\mathbb{E}_{x\sim \widehat{F}}[f(x)^2]-\mathbb{E}_{x\sim \widehat{F}}[f(x)]^2)\leq 2\mathbb{E}_{x\sim \widehat{F}}[f(x)]$. As a result, we get that there is a universal constant $C$ such that for any function class $\mathcal{F}:\mathcal{X}\to [0,1]$ with pseudo-dimension $d$, with probability $1-\delta$, we have for all $f\in \mathcal{F}$ that 
\begin{align*}
    \bigl|\mathbb{E}_{F}[f(x)]-\mathbb{E}_{\widehat{F}}[f(x)]\bigr| \leq C\left(\sqrt{\frac{\,\mathbb{E}_{\widehat{F}}[f(x)] d \log (n/\delta)}{n}} + \frac{d\log (n/\delta)}{n}\right).
\end{align*}

 Let $\Delta:=|\mathbb{E}_{\hat{F}}[f(x)-\mathbb{E}_F[f(x)]|$ and $b:=d\log (n/\delta)/n$. The above claim can be restated as $\Delta\leq C\sqrt{(\mathbb{E}_{F}[f(x)]+\Delta)\,b}+C'b$ for some constants $C,C'$.  Using $\sqrt{a+b}\leq\sqrt{a}+\sqrt{b}$ and then applying AM-GM as $C\sqrt{\Delta b}\leq\tfrac{1}{2}\Delta+\tfrac{C^2}{2}b$ gives us $\Delta\leq 2C\sqrt{\mathbb{E}_F[f(x)]\cdot b}+(C^2+2C')b$. Thus, whenever the above inequality holds, we also have for a different universal constant $C'$ that
\begin{align*}
    \bigl|\mathbb{E}_{F}[f(x)]-\mathbb{E}_{\widehat{F}}[f(x)]\bigr| \leq C'\left(\sqrt{\frac{\,\mathbb{E}_{F}[f(x)] d \log (n/\delta)}{n}} + \frac{d\log (n/\delta)}{n}\right).
\end{align*}

Let $\mathcal{G}$ be the class of function on $[0,1]$ of the form $  g_{u,v, \xi_u,\xi_v}(x):= \xi_u\mathbf{1}\{x=u\}+\mathbf{1}\{u<x<v\}+ \xi_v\mathbf{1}\{x=v\}$, where $0\leq \xi_u,\xi_v\leq 1,\ u\leq v\in[0,1]$. Let $\mathcal{F}=\{xg(x): g\in \mathcal{G}\}$. We wish to bound the pseudo-dimension of these classes. Recall that for any class $\mathcal{F}$, the pseudo-dimension is precisely the VC dimension of $\{(x,t) \mapsto \mathbf{1}\{f(x)\geq t\}: f\in \mathcal{F}\}$. We use the following result from \citet{goldberg1993bounding} on VC dimension of classes involving real numbers to conclude that the pseudo-dimension of both classes is at most a constant, by observing that membership for the subgraph sets of such function classes can be computed by a constant sized algebraic circuit that permits usual arithmetric operations and (in)equalities.

\begin{theorem}[\cite{goldberg1993bounding}]
Let $\{\mathcal{C}_{k,n} : k, n \in \mathbb{Z}_{>0}\}$ be a family of concept classes where concepts in $\mathcal{C}_{k,n}$ and inputs are represented by $k$ and $n$ real values, respectively. Further, let the test for membership of an instance $x$ in a concept $C$ in $\mathcal{C}_{k,n}$ consist of an algorithm $\mathcal{A}_{k,n}$ taking $k + n$ real inputs representing $C$ and $x$, whose runtime is $t=t(k, n)$, and which returns the truth value $a\in C$. The algorithm $\mathcal{A}_{k,n}$ is allowed to perform
conditional jumps (conditioned on equality and inequality of real values) and execute the standard arithmetic operations on real numbers $(+,-,\times,/)$ in constant time. Then the VC dimension of $\mathcal{C}_{k,n}$ is at most $O(kt)$.
\end{theorem}

Note that for any monotone partitional signaling scheme $\pi$, as detailed in \Cref{sec:conv}, $\Pr_{\pi\circ F}[k]$ and $\mathbb{E}_{\pi\circ F}[v\mathbf{1}\{\text{signal } k \text{ is observed}\}]$ can be realized as expectations of specific members in $\mathcal{G}$ and $\mathcal{F}$, respectively. Thus, we already have with high probability for a suitable universal constant $C''$ that
\begin{align*}
    &\max\left\lbrace\left|\Pr_{\pi\circ F}({k}) - \Pr_{\pi\circ \widehat{F}}({k}) \right|, \left|\mathbb{E}_{\pi\circ F}[v\mathbf{1}\{k \text{ is observed}\}] - \Pr_{\pi\circ \widehat{F}}[v\mathbf{1}\{k \text{ is observed}\}] \right|\right\rbrace\\
    &\leq C''\left(\sqrt{\min\{\Pr_{\pi\circ {F}}(k),\Pr_{\pi\circ {\widehat{F}}}({k})\}\frac{\log (n /\delta)}{n}} + \frac{\log (n /\delta)}{n}\right).
\end{align*}

To conclude the statement concerning the conditional means, we observe the following elementary inequality for positive reals:
$\left|\frac{A'}{B'}-\frac{A}{B}\right| \leq \frac{|A'-A|}{B} + \frac{A'}{B'}\frac{|B'-B|}{B}$
and substitute the appropriate quantities from the last inequality, while noting that the condition mean, i.e., $A'/B'$, can at most be one.

\subsection{Ensuring Monotonic Margins}\label{sec:rep}
Let us consider that we are given a monotone partitional scheme $(t_k:k\in [K])$ and an associated incentive compatible direct menu $((q_k, p_k): k \in [K])$ for a value distribution $\widehat{F}$, which might not be revenue maximizing pair. A key condition used in previous proofs are that the seller's margins $p_k-c(q_k)$ are non-decreasing in $k$. This is true by default for optimal incentive-compatible direct menus with respect to $\widehat{F}$; our approach here will give an alternative proof of this. However, if we use the approximately optimal menu from \Cref{sec:fptas}, this may not be true. To placate this worry, we now provide an algorithm that given an incentive-compatible menu converts it into another incentive compatible menu where margins are non-decreasing in the index, while weakly increasing the revenue. Thus, this condition can always be ensured algorithmically.

The modification works as follows: first, we calculate the margins $\pi_k = p_k-c(q_k)$ for all signals. Next, we only retain menu items in the set $S=\{k: \pi_k\geq \max_{j<k} \pi_j\}$ and delete the rest. This results in an indirect menu. So, we calculate what a buyer observing signal $k$ purchases in this smaller menu (with respect to the prior $\widehat{F}$) and designate that as their quality-item pair in a new direct menu.

\begin{proposition}
    Given any incentive-compatible direct menu, the above procedure produces another incentive-compatible direct menu with non-decreasing margins. In addition, during this transformation, the revenue weakly increases.
\end{proposition}
\begin{proof}
    The incentive compatibility holds because of the relabeling in the last step. Notice that the item corresponding to smallest $k$ is retained by definition, and hence participation for all (signal) types is upheld. Let $\sigma(k)$ be the item a buyer observing signal $k$ picks in this smaller menu $S$. It is well known incentive-compatibility for one-dimensional (signal) types ensures that $q_{\sigma(k)}$ is non-decreasing \citep{maskin1984monopoly}, and hence, by construction, $\pi_{\sigma(k)}$ follows the same pattern. As for preserving the revenue, this is immediate for any signal $k$ for which $k$ belongs to $S$, because a buyer receiving this signal will still pick the same item, which in turn has not been deleted. If $k$ is not in $S$, let $k^-$ be the largest index picked in $S$ that is smaller than $k$, and $k^+$ be the smallest index picked in $S$ that is larger than $k$. Since incentive-compatibility constraints only enforce on neighboring types even for suboptimal menus \citep{maskin1984monopoly}, $\sigma(k)$ is either $k^-$ or $k^+$. Because $k$ is not $S$, we know $\pi_{k^{-}}\geq \pi_k$, and similarly, $\pi_{k^{+}}\geq \pi_k$. So, without knowing the precise details of what signal $k$ picks, we know that the margin (and hence the net revenue) is weakly greater.
\end{proof}

\subsection{Proof of Theorem~\ref{thm:justone}}
We consider $\mathcal{X}_{\texttt{s}} = \{((I,\bm{w},\bm{\xi}),(p,q))$ with $I=2$ and thus $\bm{w} = \{0=w_0\leq w_1\leq w_2=1\}$\}, $\bm{\xi} = \{0=\xi_0\leq \xi_1\leq \xi_2=1\}$\}. Thus we can equivalently write $\mathcal{X}_{\texttt{s}}=\{ (w_1,\xi_1,(p,q))\}$. Here the information structure simply consists of two signals $s_{\texttt{H}}$ and $s_{\texttt{L}}$ where 
\begin{align*}
    \pi(s_{\texttt{L}}|v) = \begin{cases}
        1 \ & \text{if}\ v<w_1\\
        \xi_1 & \text{if}\ v=w_1
    \end{cases}, \quad \pi(s_{\texttt{H}}|v) = \begin{cases}
        1 \ & \text{if}\ v>w_1\\
        1-\xi_1 & \text{if}\ v=w_1
    \end{cases}
\end{align*}
 Let $\mu_F(w_1,\xi_1)$ be the posterior mean under signal $s_{\texttt{H}}$, which is explicitly given as 
 \begin{align*}
     \mu_F(w_1,\xi_1) = \frac{\E_F[(1-\xi_1)v\bm{1}\{v=w_1\}+v
     \bm{1}\{v>w_1\}]}{\E_F[(1-\xi_1)\bm{1}\{v=w_1\}+\bm{1}\{v>w_1\}]}.
 \end{align*} 
 The optimization problem can be written as  
\begin{align*}
   \texttt{L}_{F}^{(1)}:\ \max_{\substack{p\in \mathbb{R},q\in[\qmin,\qmax]\\w_1,\xi_1\in [0,1]}}  \underbrace{\E_F[(1-\xi_1)\bm{1}\{v=w_1\}+\bm{1}\{v>w_1\}]}_{\text{probability of realizing signal}\ s_{\texttt{H}}:\ 
\Pr_{F}(s_{\texttt{H}}|w_1,\xi_1)}\cdot\underbrace{ \left( p -c(q)\right)}_{\text{revenue}} \cdot \underbrace{\mathbb{I}\left\{\mu_F(w_1,\xi_1)\cdot q- p\geq 0\right\}}_{\text{participation of the buyer}}.
\end{align*}
Given $N$ samples $v_1,\ldots,v_N\overset{\mathrm{iid}}{\sim} F$, let $\hat{F}(v) = |\{v_i:v_i\leq v\}|/n=\sum_{i=1}^n\mathbb{I}\{v_i\leq v\}/n$ and let $(w,\xi,(p,q))$ and  $(\hat{w},\hat{\xi},(\hat{p},\hat{q}))$ be the solution of $\texttt{L}^{(1)}_{F}$ and $\texttt{L}^{(1)}_{\hat{F}}$ respectively.

Define 
\begin{align*}
    \mathcal{G} = \left\{g_{w,\xi}(v) = (1-\xi)\bm{1}\{v=w\}+\bm{1}\{v>w\}\ \forall \xi\in [0,1],\ w,v\in [0,1]\right\}
\end{align*}

We have with probability at least $1-\delta$, the following from the analysis of section~\ref{appsec:conc}
\begin{align*}
    \sup_{\xi\in[0,1],w\in[0,1]}|\E_{v\sim F}[g_{w,\xi}(v)] - \E_{v\sim \hat{F}}[g_{w,\xi}(v)]\leq \varepsilon
\end{align*}
for $\varepsilon=C \sqrt{\tfrac{1}{2n}\ln\left(\tfrac{1}{\delta}\right)}$ for some constant $C>0$. Using this, a simple application of triangle inequality gives for all $\xi\in[0,1]$ and $w\in[0,1]$ with probability at least $1-\delta$, the following
$$|\mu_F(w,\xi)-\mu_{\hat{F}}(w,\xi)|\leq \min\{\rho_{w,\xi},\hat{\rho}_{w,\xi}\},$$ where $\hat{\rho}_{w,\xi}=\varepsilon/\E_{v\sim \hat{F}}[g_{w,\xi}(v)]$ and $\rho_{w,\xi} = \varepsilon/\E_{v\sim F}[g_{w,\xi}(v)]$. \\
    Thus, $\mu_{\hat{F}}(\hat{w},\hat{\xi})\cdot \hat{q}-\hat{p}\geq 0$ implies $\mu_{F}(\hat{w},\hat{\xi})\cdot \hat{q}-\hat{p}'\geq0$ where $\hat{p}'=\hat{p}-\hat{\rho}_{\hat{w},\hat{\xi}} \qmax$. \\
    Similarly, $\mu_F(w,\xi)\cdot q-p\geq 0$  implies $\mu_{\hat{F}}(w,\xi)\cdot q-p'\geq 0$ for $p'=p-{\rho}_{w,\xi} \qmax$. \\
    This simply means that $(w,\xi,(p',q))$ satisfies the \texttt{IR} conditions under $\hat{F}$ and similarly $(\hat{w},\hat{\xi},(\hat{p}',\hat{q}))$ satisfies the \texttt{IR} conditions under $F$. Now, observe that
\begin{align*}
    \texttt{Rev}(\hat{w},\hat{\xi},(\hat{p}',\hat{q});F)& =  \E_{v\sim F}[g_{\hat{w},\hat{\xi}}(v)]\cdot(\hat{p}'-c(\hat{q})) \\
    &\geq \left(\E_{v\sim \hat{F}}[g_{\hat{w},\hat{\xi}}(v)]-\varepsilon\right)\cdot\left(\hat{p}-\rho_{\hat{w},\hat{\xi}}\qmax-c(\hat{q})\right)\\
    &\geq \left(\E_{v\sim \hat{F}}[g_{\hat{w},\hat{\xi}}(v)]\right)\cdot\left(\hat{p}-
    c(\hat{q})\right)-\E_{v\sim \hat{F}}[g_{\hat{w},\hat{\xi}}(v)]\cdot\rho_{\hat{w},\hat{\xi}}\qmax-\varepsilon \qmax\\
    &=  \texttt{Rev}(\hat{w},\hat{\xi},(\hat{p},\hat{q});\hat{F}) - 2\varepsilon \qmax
    \end{align*}
Using the feasibility of $(w,\xi,(p',q))$ for $\texttt{L}_{\hat{F}}$, we get
\begin{align*}
  \texttt{Rev}(\hat{w},\hat{\xi},(\hat{p}',\hat{q});F)&\geq   \texttt{Rev}(\hat{w},\hat{\xi},(\hat{p},\hat{q});\hat{F}) - 2\varepsilon \qmax
  \\& \geq  \texttt{Rev}(w,\xi,(p',q);\hat{F}) -2\varepsilon \qmax\\
  &= \E_{v\sim \hat{F}}[g_{w,\xi}(v)]\cdot(p'-c(q))-2\varepsilon \qmax\\
       &\geq  \left(
       \E_{v\sim F}[g_{w,\xi}(v)]-\varepsilon\right)\cdot\left(p-\rho_{w,\xi}\qmax-c(q)\right)-2\varepsilon \qmax \\
       &\geq  
       \E_{v\sim F}[g_{w,\xi}(v)]\cdot\left(p-c(q)\right)-\varepsilon \qmax-\E_{v\sim F}[g_{w,\xi}(v)]\cdot \rho_{w,\xi}\qmax-2\varepsilon \qmax \\
        &\geq \texttt{Rev}(w,\xi,(p,q);F) - 4\varepsilon \qmax
\end{align*}
 Thus with probability at least $1-\delta$ for $\varepsilon$ additive error we need $\mathcal{O}(\tfrac{\qmax^2}{\varepsilon^2}\log \tfrac{1}{\delta})$ samples.

\section{Sample Compelxity with Demand Queries}\label{sec:appdem}
\subsection{Proof of Theorem~\ref{thm:notdisc}}
We already know from \Cref{lemma:momentrecovery} that $m_j$ captures the $j^{th}$ moment exactly. Let $\tilde{F}_j=\sum_{i=0}^j \alpha_i L_i(x)$ be the degree-$j$ Legendre projection of $F$. Since $L_j$'s are a family of orthogonal polynomials under the usual measure, we know \begin{align*}
    \alpha_j &= \int_0^1 F(x) L_j(x)dx = \sum_{i=0}^j a_{ji} \int_0^1 F(x)x^{i} dx \\
    &= \sum_{i=0}^j a_{ji} \left(1-\int_0^1 \frac{x^{i+1}}{(i+1)} f(x) dx\right) = \sum_{i=0}^j a_{ji} \left(1-\frac{m_{i+1}}{i+1}\right),
\end{align*}
as claimed. For the de la Vallee Poussin mean, which are suffix averages of higher-degree Legendre projections, using Theorem 2.1 from \citet{mastroianni2008vallee}, with $\alpha=\beta=\delta=\gamma=\nu=0$ and $p=\infty$ substituted in their statement, for a constant $C$ that changes neither with $n$ nor $F$, we have $$
\max_{v\in [0,1]}|F(v) - \hat{F}_{\lceil n/2\rceil} (x)| \leq {\Delta} :={\frac{CB^r}{(n/2)^r}}.
$$

Now it follows that $\max_{v\in [0,1]} |F(v)-\widehat{F}(v)|\leq \Delta$, because forcing the CDF estimate to be non-negative and non-decreasing can only (weakly) decrease the $L_\infty$ distance. Note that in this section, both $F$ and $\widehat{F}$ are continuous, the former by assumption, and the latter because it is a polynomial. For continuous distributions, signals are simple partitions of the value space, as explained in \Cref{sec:conv}. Consider any generic signal $k$ such that $\pi(k|v) = \mathbf{1}\{v\in [v_0, v_1)\}$. Now, we have that
\begin{align*}\Pr_{\pi\circ F}[k] &= F(v_1)-F(v_0),\\ \mathbb{E}_F[v\mathbf{1}\{v\in [v_0, v_1)\}] &= \int_{v_0}^{v_1} v f(v) dv = v_1F(v_1) - v_0F(v_0) - \int_{v_0}^{v_1} F(v) dv.
\end{align*}
Thus, we can transfer $L_\infty$ approximation guarantees from $F$ to probability of occurrence of signals, and the corresponding posterior means. Concretely, we have \begin{align*}
    \left|\Pr_{\pi\circ F}({k}) - \Pr_{\pi\circ \widehat{F}}({k}) \right| &\leq \Delta,  \text{ and }  
     \left|\mathbb{E}_{F}[v|k] - \mathbb{E}_{\widehat{F}} [v|k] \right| \leq \frac{\Delta}{\max\{\Pr_{\pi\circ F}({k}), \Pr_{\pi\circ \widehat{F}}({k})\}}, 
\end{align*}
where the last display uses the inequality that
$\left|\frac{A'}{B'}-\frac{A}{B}\right| \leq \frac{|A'-A|}{B} + \frac{A'}{B'}\frac{|B'-B|}{B}$.

From here on, we essentially retrace the proof of \Cref{thm:eps3samplecomplexity}. We decompose the error as 
    \begin{align*}
          Rev^*(F) -   Rev(\tilde{x}_{\hat{F}};F) & =  \underbrace{Rev^*(F) - Rev^*(\hat{F})}_{\texttt{Term1}}  + \underbrace{Rev^*(\hat{F})- Rev(\tilde{x}_{\hat{F}};F,F)}_{\texttt{Term2}}.
    \end{align*}

To upper bound \texttt{Term 1}, let $\bm{x}_F$ be the optimal pair of menu $((q_k, p_k): k\in [K])$ and monotone partitional signaling scheme $(t_k: k\in [K])$ under $F$. Let $\bm{z}_{\tau}=\{k\in [K]: t_k-t_{k-1}< \tau\}$. Define $\tilde{\bm{x}}_F$ with signaling scheme $\tilde{\pi}$ modified so that all signals in $\bm{z}_\tau$ are mapped to a \texttt{null} signal,
and the prices discounted as $\tilde{p}_k = p_k-k\rho$. Then, we have that
    \begin{align*}
        \texttt{Term1}&:= Rev^*(F) - {Rev}(\tilde{x}_F;\hat{F},\hat{F}) + {Rev}(\tilde{x}_F;\hat{F},\hat{F}) - Rev^*(\hat{F}) \\
        &\leq Rev^*(F) - {Rev}(\tilde{x}_F;\hat{F},\hat{F})
    \end{align*}
    where we once again appeal to the fact that $Rev^*(\hat{F})$ is the optimal revenue for the empirical distribution $\widehat{F}$. Now, we use \Cref{lemma:samplecomplexityimpA}, which places an upper bound of ${\mathcal{O}}(\overline{q}K\sqrt{\Delta})$ on \texttt{Term 1}, by choosing $\rho=4\overline{q}\sqrt{\Delta}$ and $\tau=\sqrt{\Delta}$.

Similarly, for $\texttt{Term2}$, let $\bm{x}_{\widehat{F}}$ be the optimal pair of menu $((q_k, p_k): k\in [K])$ and monotone partitional signaling scheme $(t_k: k\in [K])$ this time under $\widehat{F}$. Let $\bm{\tilde{z}}_{\tau}=\{k\in [K]: t_k-t_{k-1}< \tau\}$. Let $\bm{\tilde{x}}_{\widehat{F}}$ be the analogous modification where all signals in $\bm{\tilde{z}}_{\tau}$ are mapped to a null signal, and prices are linear-additively discounted by $\rho$ at all indices. Now, applying \Cref{lemma:samplecomplexityimpA}, with the roles of $F$ and $\widehat{F}$ reversed, we get 
    \begin{align*}
        \texttt{Term2}&:= Rev^*({\widehat{F}}) - {Rev^*}(\bm{x}_{\widehat{F}};{\widehat{F}},{\widehat{F}}) + {Rev}(\bm{x}_{\widehat{F}};{\widehat{F}},{\widehat{F}}) - Rev^*(\bm{\tilde{\bm{x}}}_{\widehat{F}}; F,F) \\
        &\leq Rev(\bm{x}_{\widehat{F}};{\widehat{F}},{\widehat{F}}) - {Rev}(\tilde{\bm{x}}_F;\hat{F},\hat{F}) ={\mathcal{O}}(\overline{q}K\sqrt{\Delta}).
    \end{align*}
Now, we solve for $n$ that makes the net revenue loss $\varepsilon$.

\section{Regret under Joint Learning}\label{sec:appreg}
Define the following quantity,
\begin{align*}
    \texttt{GAP}(\bm{x},F,\mathcal{C}(F)) := \sup_{G\in \mathcal{C}(F)}\left(\underbrace{Rev(\bm{x}^\star_F;F,F)}_{\substack{\text{optimal revenue under the}\\\text{common prior of $F$}}}-\underbrace{Rev(\bm{x};G,F)}_{\substack{\text{revenue under $\bm{x}$}\\\text{when buyer holds prior $G$}}} \right). 
\end{align*}
For any CDF $F$ over $[0,\vmax]$ with $\vmax\leq 1$, let the solution of $\texttt{L}_{F}$ be denoted by $\bm{x}_F=((K,\{t_i\}_{i=0}^K),(p_i,q_i)_{i=1}^K)$ where $0=t_0<\ldots<t_K=1$ is the quantile partition of the motonone partitional information structure, and $(p_i,q_i)$ is the \textit{direct} menu item corresponding to the quantile interval $(t_{k-1},t_k]$. For $\tau\in [0,1]$ and $\rho\geq  0$, define the transformation $\tilde{\bm{x}}_F=\mathsf{Tr}(\bm{x}_F;\tau,\rho)$ where $\tilde{\bm{x}}_F$ consists of the following modified signaling scheme in the value-space for $\bm{z}_{\tau}=\{k\in [K]:t_k-t_{k-1}<\tau\}$ 
\begin{align*}
    \tilde{\pi}(s\mid v) = \begin{cases} \pi(s_k\mid v) & \text{if } s=s_k,\ k\notin\bm{z}_\tau,\ \forall v\in [0,\vmax]\\ \sum_{k\in\bm{z}_\tau}\pi(s_k\mid v) & \text{if } s=\texttt{null}, \ \forall v\in [0,\vmax]\\
    1 & \text{if } s=\texttt{null}, \ \forall v\in (\vmax,1]\end{cases}  
\end{align*}
where $(\pi(k\mid v):\forall v\in[0,\vmax], k\in [K])$ is the value-space description w.r.t $F$ from  \Cref{sec:conv} with a modified menu $(\tilde{p}_i,q_i)_{i=0}^K$ where $\tilde{p}_i=p_i-i\rho\ \forall i\in[K]$.
\begin{lemma} \label{lemma:impsamplecomplexity}
    Let $F$ be a CDF over $\mathcal{V}=[0,\vmax]$ with $\vmax\leq 1$ and let $\bm{x}_F=((K,\{t_i\}_{i=0}^K),(p_i,q_i)_{i=1}^K)$ be the solution of $\texttt{L}_{F}$. For each CDF $G$, define $\E_G[v|k]$ the posterior mean of signal $s_k$ (which corresponds to the quantile interval $(t_{k-1},t_k]$ of $F$) under $G$.  Let $\mathcal{C}(F)$ be the family of CDFs $G$ on $[0,1]$ such that  we have $\left|\E_F[v|k]-\E_G[v|k]\right|\leq \rho/\qmax \quad  \forall \ k\in[K]\backslash\bm{z}_{\tau}$ for $\bm{z}_{\tau}=\{k\in[K]:t_k-t_{k-1}<\tau\}$.
    Then for $\bm{\tilde{x}}_F = \mathsf{Tr}(\bm{x}_F;\tau,\rho)$, we have $ \texttt{GAP}(\tilde{x}_F,F,\mathcal{C}(F)) \leq K(\rho + \Bar{q}\tau)$. 
\label{lemma:samplecomplexityimp}
\end{lemma}

\begin{proof}[Proof of Lemma~\ref{lemma:samplecomplexityimp}]
Since $(p_k,q_k)$ is a solution of $\texttt{L}_{F}$, it satisfies individual rationality (IR) under $F$, i.e. $E_F[v|k]\cdot q_k-p_k \geq 0$ for all $k\in [K]$.  For any $k\in[K]\backslash\bm{z}_{\tau}$, since $|\E_F[v|k]-\E_G[v|k]|\leq\rho/ \qmax$, we have under $G$ that \texttt{IR} for all $k\in[K]\backslash\bm{z}_{\tau}$ with the modified prices.

Furthermore, the incentive compatibility constraint under $F$ gives, for all $k,j\in[K]$:
\begin{align}
       \E_F[v|k]\cdot q_k-p_k  \geq   \E_F[v|k]\cdot q_j-p_j. \label{eq:lemma1ic}
\end{align}
Observe that for all $k\in[K]$ we have that 
\begin{align*}
q_k\cdot \mathbb{E}_{G}[v|k] -p_k &\geq q_k\cdot \mathbb{E}_{F}[v|k] - p_k -  \rho \\
&\geq \max\{0, \max_j \{q_j\cdot \mathbb{E}_{F}[v|k] - p_j\}\} -  \rho\\
&\geq \max\{0, \max_j \{q_j\cdot \mathbb{E}_{G}[v|k] - p_j\}\} - 2 \rho
\end{align*}

Now observe that for all $j<k$ that 
\begin{align*}
    q_k\cdot \mathbb{E}_{G}[v|k] -\tilde{p}_k &= q_k\cdot \mathbb{E}_{G}[v|k] -p_k + k\rho\\
    &\geq q_j\cdot \mathbb{E}_{G}[v|k] -p_j + k\rho -\rho\\
    &= q_j\cdot \mathbb{E}_{G}[v|k] -\tilde{p}_j + (k-j)\rho -\rho\\
    &\geq q_j\cdot \mathbb{E}_{G}[v|k] -\tilde{p}_j.
\end{align*}

Thus $\texttt{GAP}(\tilde{x}_F,F,C(F))\leq K(\rho + \Bar{q}\tau)$ using the same argument as in Lemma~\ref{proof:lemmaimpA}.
\end{proof}

\begin{definition}\label{def:G}
Let $\mathcal{G}$ denote the class of functions $g_{u,v,\xi_u,\xi_v}:[0,\vmax]\to[0,1]$ of the form
\[
    g_{u,v,\xi_u,\xi_v}(x) := \xi_u\mathbf{1}\{x=u\}+\mathbf{1}\{u<x<v\}+\xi_v\mathbf{1}\{x=v\},\quad u\leq v\in[0,\vmax],\quad \xi_u,\xi_v\in[0,1].
\]
\end{definition}

\begin{lemma}\label{lemma:posteriormeanCt}
Let $F,G\in\Delta([0,1])$ be CDFs with $\sup_v |F(v)-G(v)|\leq \varepsilon$. For any
$g\in\mathcal{G}$ such that $\E_F[g],\E_G[g]>0$, we have
\begin{align*}
    \left|\frac{\E_F[xg]}{\E_F[g]}-\frac{\E_G[xg]}{\E_G[g]}\right|
    \leq 4\varepsilon \min\left\{\frac{1}{\E_F[g]},\frac{1}{\E_G[g]}\right\}.
\end{align*}
\end{lemma}

\begin{proof}
By triangle inequality 
\begin{align*}
\left|\frac{\E_F[xg]}{\E_F[g]}-\frac{\E_G[xg]}{\E_G[g]}\right|
&\le \frac{|\E_F[xg]-\E_G[xg]|}{\E_F[g]}
   + \frac{\E_G[xg]}{\E_F[g]}\cdot \frac{|\E_F[g]-\E_G[g]|}{\E_G[g]}.
\end{align*}
Since $\E_G[xg]/\E_G[g]\leq 1$ as $x\in[0,1]$, it suffices to bound
$|\E_F[h]-\E_G[h]|$ for $h\in\{g,xg\}$.
Observe that
\begin{align*}
    |\E_F[h]-\E_G[h]|\leq d_{\mathrm{TV}}(h)\sup_v|F(v)-G(v)|\leq \varepsilon d_{\mathrm{TV}}(h)
\end{align*}
where $d_{\mathrm{TV}}(h)=\sup_n\sup_{0=x_0<x_1<\ldots<x_n=1}\sum_{i=1}^n|h(x_i)-h(x_{i-1})|$ is the total variation of the function $h$ on $[0,1]$, this is so since
\begin{align*}
    \E_F[h]-\E_G[h]=\int_{0}^1h\mathrm{d}(F-G)=-\int_{0}^1 (F-G)\mathrm{d}h.
\end{align*}
Now $g\in\mathcal G$ has the total variation at most $2$, and $x\cdot g$ has the total
variation at most $2$ since $x\leq 1$. Hence $|\E_F[g]-\E_G[g]|\leq 2\varepsilon$ and $
|\E_F[xg]-\E_G[xg]|\leq 2\vmax\varepsilon.$
Substituting gives us the bound. 
Swapping $F$ and $G$ yields the same bound as $\E_G[g]$ in the denominator,
and taking the minimum proves the claim.
\end{proof}

\begin{theorem}
    \Cref{algo:online}
    with cutoff $\tau_t$ and  discount $\rho_t = 4\Bar{q}\varepsilon_t/\tau_t$
    is $(\beta,\bm{\kappa})$-robustly revenue efficient for sets $\mathcal{C}(\hat{F}_t)=\{G\in \Delta
    (\mathcal{V}):\sup_v|G(v)-\hat{F}_t(v)|\leq \varepsilon_t
    \}$ for $\beta\geq \sum_{t\in[T]}2\exp\left(-t\varepsilon_t^2/2\right)$ and $\kappa_t = K\Bar{q}(4\varepsilon_t/\tau_t + \tau_t)$ where $\mathcal{V}=[0,1]$ for any sequence $\{\varepsilon_t\}_{t=T},\{\tau_t\}_{t=1}^T$ with $\tau_t,\varepsilon_t\in (0,1]\ \forall t\in [T]$.
    \label{thm:robustefficiency}
\end{theorem}
\begin{proof}[Proof of Theorem~\ref{thm:robustefficiency}]
    
Consider $G\in \mathcal{C}(\hat{F}_t)$, then we have by definition $ \sup_{v\in\mathcal{V}}|G(v) - \hat{F}_t(v) |\leq \varepsilon_t$.

From Lemma~\ref{lemma:posteriormeanCt}, we have $\forall G\in\mathcal{C}(\hat{F}_t)$ the following for all $[u,v]\subseteq \mathcal{V}$
     \begin{align*}
    \left|\frac{\E_F[xg]}{\E_F[g]}-\frac{\E_G[xg]}{\E_G[g]}\right|
    \leq 4\varepsilon \min\left\{\frac{1}{\E_F[g]},\frac{1}{\E_G[g]}\right\}.
\end{align*}
     Let $\bm{x}_{\hat{F}}=\left((I,\bm{t}),(p_i,q_i)_{i\in [I]})\right)$ be the solution of $\texttt{L}_{\hat{F}_t}$ and let  $\bm{z}_{\tau_t}=\{i\in [I]: t_i-t_{i-1}< \tau_t\}\subseteq [I]$. Then from above we have $\forall G\in\mathcal{C}(\hat{F}_t) $ the following
       \begin{align*}
 \left| \E_G[v|k]-\E_{\hat{F}_t}[v|k]
\right| &\leq 4\varepsilon_t/\tau_t=:\rho_t/\Bar{q} \quad  \forall \ i\in[I]\backslash\bm{z}_{\tau}.
    \end{align*}
    For $\tilde{x}_{\hat{F}_t}=\mathsf{Tr}(\bm{x}_{\hat{F}_t};\tau_t,\rho_t)$ using Lemma~\ref{lemma:samplecomplexityimp}, we get 
    \begin{align*}
        \texttt{GAP}(\tilde{x}_{\hat{F}_t},\hat{F}_t,\mathcal{C}(\hat{F}_t)) \leq K\Bar{q}(4\varepsilon_t/\tau_t + \tau_t)=:\kappa_t.
    \end{align*}
Finally,
    \begin{align*}
    \text{Pr}_{F^\star}\left(\cap_{t=1}^T\mathcal{C}(\hat{F}_t)\not\ni F^\star\right) &=    \text{Pr}_{F^\star}\left(\cup_{t=1}^T\mathcal{C}(\hat{F}_t)^c\ni F^\star\right)  \\
    &\leq \sum_{t\in[T]}\text{Pr}_{F^\star}\left(\mathcal{C}(\hat{F}_t)^c\ni F^\star\right) \\
    &= \sum_{t\in[T]}\text{Pr}_{F^\star}\left(\sup_{v}|F^\star(v) - \hat{F}_t(v)|> \epsilon_t\right) \\
    &\leq \sum_{t\in[T]}2\exp\left(-2t\epsilon_t^2\right),
\end{align*}

where the last inequality follows from the Dvoretzky–Kiefer–Wolfowitz (DKW) \cite{Dvoretzky1956asymptotic,massart1990tight} inequality  given as follows
\begin{align*}
    \text{Pr} \left(\sup_{v\in \mathbb{R}} |F(v)-\hat{F}(v)|>\epsilon_{n}\right)\leq 2\exp\left(-2n\epsilon_n^2\right).
\end{align*}
\end{proof}

\begin{theorem}  \Cref{algo:online}
    with cutoff $\tau_t=\sqrt{\varepsilon_t}$ and  discount $\rho_t = 4\Bar{q}\sqrt{\varepsilon_t}$, for $\mathcal{C}(\hat{F}_t)=\{G\in \Delta
    (\mathcal{V}):\sup_v|G(v)-\hat{F}_t(v)|\leq \varepsilon_t
    \}$ with $\varepsilon_t=\sqrt{3\log T/ t}$ gives the following cumulative regret $\texttt{Reg}
    \leq    \mathcal{O}\left( K\bar{q}(\log T)^{1/4}T^{3/4}\right)$ with probability at least $1-3T^{-1/2}$.
\label{thm:regretfirst}
\end{theorem}

\begin{proof}[Proof of Theorem~\ref{thm:regretfirst}]
Define the event $\mathcal{E}:=\cap_{t=1}^T\{F^\star\in\mathcal{C}(\hat{F}_t)\}$. 
Let $\bm{x}_{\hat{F}}=\left((I,\bm{t}),(p_i,q_i)_{i\in [I]})\right)$ be the solution of $\texttt{L}_{\hat{F}_t}$ and let  $\bm{z}_{\tau_t}=\{i\in [I]: t_i-t_{i-1}< \tau_t\}\subseteq [I]$. Then recall that \Cref{algo:online} uses $\tilde{x}_{\hat{F}_t}=\mathsf{Tr}(\bm{x}_{\hat{F}_t};\tau_t,\rho_t)$.
Then for $\varepsilon_t=\sqrt{3\log T/t}$ and $\tau_t=\sqrt{\varepsilon_t}\ \forall t\in[T]$, \Cref{thm:robustefficiency} gives $(\beta,\bm{\kappa})$-robustly revenue efficiency for $\beta=\sum_{t=1}^T 2\exp(-t\varepsilon_t^2/2)=\sum_{t=1}^T 2T^{-3/2}=2T^{-1/2}$ and $\kappa_t=K\Bar{q}(4\varepsilon_t/\tau_t + \tau_t) = 5K\Bar{q}\sqrt{\varepsilon_t}$. This implies $\Pr(\mathcal{E})\geq 1-2T^{-1/2}$.

Then $\forall G\in\mathcal{C}(\hat{F}_t) $ we have the following from \Cref{lemma:posteriormeanCt}
       \begin{align*}
 \left| \E_G[v|k]-\E_{\hat{F}_t}[v|k]
\right| &\leq 4\varepsilon_t/\tau_t=:\rho_t/\Bar{q} \quad  \forall \ i\in[I]\backslash\bm{z}_{\tau},
    \end{align*}
    and using Lemma~\ref{lemma:samplecomplexityimp}, we get 
    \begin{align*}
        \texttt{GAP}(\tilde{x}_{\hat{F}_t},\hat{F}_t,\mathcal{C}(\hat{F}_t)) \leq K\Bar{q}(4\varepsilon_t/\tau_t + \tau_t)=\kappa_t.
    \end{align*}

\textit{Regret upper bound}: 
On $\mathcal{E}$, $\sup_v|F^\star(v)-\hat{F}_t(v)|\leq\varepsilon_t\ \forall t\in [T]$. Also $F^b_t\in\mathcal{C}(\hat{F}_t)$ by Assumption~1.  Decompose the $t^{\text{th}}$-regret term as follows
\begin{align*}
    \texttt{OPT}(F^\star) - Rev(\tilde{\bm{x}}_{\hat{F}_t};F^b_t,F^\star)
    &= \underbrace{\bigl[\texttt{OPT}(F^\star)-\texttt{OPT}(\hat{F}_t)\bigr]}_{\mathrm{(A)}}
    +\underbrace{\bigl[\texttt{OPT}(\hat{F}_t)-Rev(\tilde{\bm{x}}_{\hat{F}_t};F^b_t,\hat{F}_t)\bigr]}_{\mathrm{(B)}}\\
    &+\underbrace{\bigl[Rev(\tilde{\bm{x}}_{\hat{F}_t};F^b_t,\hat{F}_t)-Rev(\tilde{\bm{x}}_{\hat{F}_t};F^b_t,F^\star)\bigr]}_{\mathrm{(C)}}.
\end{align*}

 Since $F^b_t\in\mathcal{C}(\hat{F}_t)$  we get $\mathrm{(B)}\leq \texttt{GAP}(\tilde{\bm{x}}_{\hat{F}_t},\hat{F}_t,\mathcal{C}(\hat{F}_t))\leq \kappa_t.         $\\
 For $\mathrm{(C)}$ the buyer's decision is the same, so it can be upper bounded by $\Bar{q}\sup_{v}|F^\star(v)-\hat{F}_t(v)|\sum_{k=1}^Kd_{\mathrm{TV}}(g_k)\leq 2\Bar{q}\varepsilon_tK$ since $g$ has 2 jumps of size at most 1, where $d_{\mathrm{TV}}(\cdot)$ is the total variation distance.

Define $\tilde{\bm{x}}_{F^\star}=\mathsf{Tr}(\bm{x}_{F^\star};\tau_t,\rho_t)$ for $\bm{x}_{F^\star}=((K^\star,\bm{t}^\star),(p_i^\star,q_i^\star)_{i=1}^K)$ solution of $\texttt{L}_{F^\star}$. Then $\forall G\in\mathcal{C}(F^\star) $ we have the following from \Cref{lemma:posteriormeanCt}
       \begin{align*}
 \left| \E_G[v|k]-\E_{F^\star}[v|k]
\right| &\leq 4\varepsilon_t/\tau_t=\rho_t/\Bar{q} \quad  \forall \ k\in[K^\star]: t^\star_k-t^\star_{k-1}\geq \tau_t,
    \end{align*}
    and using Lemma~\ref{lemma:samplecomplexityimp}, we get 
    \begin{align*}
        \texttt{GAP}(\tilde{\bm{x}}_{F^\star},F^\star,\mathcal{C}(F^\star)) \leq K\Bar{q}(4\varepsilon_t/\tau_t + \tau_t)=\kappa_t.
    \end{align*}
   Recall that under $\mathcal{E}$, we also have $\hat{F}_t\in \mathcal{C}(F^\star)\ \forall t\in[T]$. 
    For (A) we perform the following decomposition
\begin{align*}
    \mathrm{(A)}
    &\leq\underbrace{\bigl[\texttt{OPT}(F^\star)-Rev(\tilde{\bm{x}}_{F^\star};\hat{F}_t,F^\star)\bigr]}_{\leq\,\texttt{GAP}(\tilde{\bm{x}}_{F^\star},F^\star,\mathcal{C}(F^\star))\,\leq\,\kappa_t}
    +\underbrace{\bigl[Rev(\tilde{\bm{x}}_{F^\star};\hat{F}_t,F^\star)-Rev(\tilde{\bm{x}}_{F^\star};\hat{F}_t,\hat{F}_t)\bigr]}_{\leq 2\Bar{q}\varepsilon_tK\ \text{ using same TV argument as (C)}}\\
    &+\underbrace{\bigl[Rev(\tilde{\bm{x}}_{F^\star};\hat{F}_t,\hat{F}_t)-\texttt{OPT}(\hat{F}_t)\bigr]}_{\leq\,0\ 
    \text{due to suboptimality of }\tilde{\bm{x}}_{F^\star}\text{ under }\hat{F}_t}\\
    &\leq\kappa_t+2\Bar{q}\varepsilon_tK.
\end{align*}
Combining we get overall 
\begin{align*}
    \texttt{OPT}(F^\star)-Rev(\tilde{\bm{x}}_{\hat{F}_t};F^b_t,F^\star)\leq 2\kappa_t+4\Bar{q}\varepsilon_tK.
\end{align*}

Finally, define $Y_t:=\texttt{rev}_t(\tilde{x}_{\hat{F}_t};F^b_t)-Rev(\tilde{\bm{x}}_{\hat{F}_t};F^b_t,F^\star)$. Since $\tilde{\bm{x}}_{\hat{F}_t}$ and $F^b_t$ are $\mathcal{H}_t$-measurable and $v_t|\mathcal{H}_t\sim F^\star$, we have $\E[Y_t|\mathcal{H}_t]=0$ and $|Y_t|\leq \Bar{q}$ where $\mathcal{H}_t$ denotes the $\sigma$-algebra of all the histories upto time $t$.  By Azuma-Hoeffding, we get
\begin{align*}
    \text{Pr}\left(\sum_{t=1}^T(-Y_t)\geq z\Bar{q}\right)\leq\exp\!\left(-\frac{z^2}{2T}\right).
\end{align*}
Setting $z=\sqrt{T\log T}$ gives $T^{-1/2}$ upper bound.

Thus overall on the event $\mathcal{E}\cap\{\sum_t(-Y_t)\leq z\Bar{q}\}$, which now holds with probability $\geq 1-2T^{-1/2}-T^{-1/2}=1-3T^{-1/2}$, we get 
\begin{align*}
    \texttt{Reg}
    =\sum_{t=1}^T\bigl(\texttt{OPT}(F^\star)-\texttt{rev}_t\bigr)
    &=\sum_{t=1}^T\bigl(\texttt{OPT}(F^\star)-Rev(\tilde{\bm{x}}_{\hat{F}_t};F^b_t,F^\star)\bigr)+\sum_{t=1}^T(-Y_t)\\
    &\leq\sum_{t=1}^T(2\kappa_t+4\Bar{q}\varepsilon_tK)+\Bar{q}\sqrt{T\log T}.
\end{align*}

Substituting $\tau_t=\sqrt{\varepsilon_t}$ gives us $\kappa_t =K\Bar{q}(4\varepsilon_t/\tau_t + \tau_t) = 5K\Bar{q}\sqrt{\varepsilon_t} $ and since $\varepsilon_t=\sqrt{3\log T/t}$, we have $\sqrt{\varepsilon_t}=(3\log T)^{1/4}t^{-1/4}$ and $\varepsilon_t=(3\log T)^{1/2}t^{-1/2}$.  Using $\sum_{t=1}^T t^{-1/4}\leq\frac{4}{3}T^{3/4}$ and $\sum_{t=1}^T t^{-1/2}\leq 2T^{1/2}$ we get the leading term as follows 
\begin{align*}
\sum_{t=1}^T\sqrt{\varepsilon_t}&\leq\tfrac{4}{3}(3\log T)^{1/4}\,T^{3/4}=\mathcal{O}\left((\log T)^{1/4}T^{3/4}\right)
\end{align*}
Finally we have
\begin{align*}
    \texttt{Reg}\leq \mathcal{O}\left(K\Bar{q}(\log T)^{1/4}T^{3/4}\right).
\end{align*}
\end{proof}

Consider the following set for which we show the $\mathcal{O}(T^{2/3})$ regret guaranty. 
\begin{align*}
    \mathcal{C}'(\hat{F}_t) = \left\{ G\in\Delta(\mathcal{V}): \left|\frac{\E_{\hat{F}_t}[xg]}{\E_{\hat{F}_t}[g]}-\frac{\E_G[xg]}{\E_G[g]}\right|\leq \varepsilon_t\sqrt{\frac{1}{\max\{\E_{\hat{F}_t}[g],\E_{G}[g]\}}}+\varepsilon_t^2\frac{1}{{\max\{\E_{\hat{F}_t}[g],\E_{G}[g]\}}}\right\}.
\end{align*}

\begin{theorem}
     \Cref{algo:online}
    with cutoff $\tau_t$ and  discount $\rho_t = \Bar{q}(\varepsilon_t/\sqrt{\tau_t}+\varepsilon_t^2/\tau_t)$
    is $(\beta,\bm{\kappa})$-robustly revenue efficient for the sets $\mathcal{C}'(\hat{F}_t)$ for $\beta\geq \sum_{t\in[T]}t\exp\left(-t\varepsilon_t^2/C\right)$ and $\kappa_t = K\Bar{q}(\varepsilon_t/\sqrt{\tau_t}+\varepsilon_t^2/\tau_t + \tau_t)$ where $\mathcal{V}=[0,1]$ for any sequence $\{\varepsilon_t\}_{t=T},\{\tau_t\}_{t=1}^T$ with $\tau_t,\varepsilon_t\in (0,1]\ \forall t\in [T]$. 
    \label{thm:robustefficiency2}
\end{theorem}
\begin{proof}[Proof of Theorem~\ref{thm:robustefficiency2}]
    
Consider $G\in \mathcal{C}'(\hat{F}_t)$.   Let $\bm{x}_{\hat{F}}=\left((I,\bm{t}),(p_i,q_i)_{i\in [I]})\right)$ be the solution of $\texttt{L}_{\hat{F}_t}$. Then we have $\forall G\in\mathcal{C}'(\hat{F}_t)$ the following by definition 
       \begin{align*}
 \left| \E_G[v|k]-\E_{\hat{F}_t}[v|k]
\right| &\leq \varepsilon_t/\sqrt{\tau_t}+\varepsilon_t^2/\tau_t=:\rho_t/\bar{q} \forall \ i\in[I]\backslash\bm{z}_{\tau}.
    \end{align*}
    For $\tilde{x}_{\hat{F}_t}=\mathsf{Tr}(\bm{x}_{\hat{F}_t};\tau_t,\rho_t)$ using Lemma~\ref{lemma:samplecomplexityimp}, we get 
    \begin{align*}
        \texttt{GAP}(\tilde{x}_{\hat{F}_t},\hat{F}_t,\mathcal{C}'(\hat{F}_t)) \leq K\Bar{q}(\varepsilon_t/\sqrt{\tau_t}+\varepsilon_t^2/\tau_t + \tau_t)=:\kappa_t.
    \end{align*}

    Finally, \begin{align*}\text{Pr}_{F^\star}\left(\cap_{t=1}^T\mathcal{C}'(\hat{F}_t)\not\ni F^\star\right) &=    \text{Pr}_{F^\star}\left(\cup_{t=1}^T\mathcal{C}'(\hat{F}_t)^c\ni F^\star\right)  \\
    &\leq \sum_{t\in[T]}\text{Pr}_{F^\star}\left(\mathcal{C}'(\hat{F}_t)^c\ni F^\star\right) \\
    &\leq \sum_{t\in[T]}t\exp\left(-t\varepsilon_t^2/C\right),
\end{align*}

where the last inequality follows from \Cref{lem:conc} for some constant $C$.
\end{proof}

\begin{theorem}  \Cref{algo:online}
    with cutoff $\tau_t=\varepsilon_t^{1/3}$ and  discount $\rho_t = \Bar{q}(\varepsilon_t/\sqrt{\tau_t}+\varepsilon_t^2/\tau_t)$, for $\mathcal{C}'(\hat{F}_t)$ with $\varepsilon_t=\sqrt{5\log T/(2t)}$ gives the following cumulative regret $\texttt{Reg}
    \leq    \mathcal{O}\left(K\Bar{q}(\log T)^{1/3}T^{2/3}\right)$ with probability at least $1-\mathcal{O}(T^{-1/2})$.
\label{thm:regretsecond}
\end{theorem}

\begin{proof}[Proof of Theorem~\ref{thm:regretsecond}]

Define the event $\mathcal{E}:=\cap_{t=1}^T\{F^\star\in\mathcal{C}(\hat{F}_t):=\{G\in\Delta([0,1]):\max_v|G(v)-\hat{F}_t(v)|\leq \xi_t\}\}$ and $\mathcal{E}':=\cap_{t=1}^T\{F^\star\in\mathcal{C}'(\hat{F}_t)\}$. 
Let $\bm{x}_{\hat{F}}=\left((I,\bm{t}),(p_i,q_i)_{i\in [I]})\right)$ be the solution of $\texttt{L}_{\hat{F}_t}$ and let  $\bm{z}_{\tau_t}=\{i\in [I]: t_i-t_{i-1}< \tau_t\}\subseteq [I]$. Then recall that \Cref{algo:online} uses $\tilde{x}_{\hat{F}_t}=\mathsf{Tr}(\bm{x}_{\hat{F}_t};\tau_t,\rho_t)$. Then for $\varepsilon_t=\sqrt{5\log T/(2t)}$ and $\tau_t=\varepsilon_t^{1/3}$, Theorem~\ref{thm:robustefficiency2} gives $(\beta,\bm{\kappa})$-robustly revenue efficiency for $\beta=\sum_{t\in[T]}t\exp\left(-t\varepsilon_t/C\right)=\mathcal{O}(T^{-1/2})$ and $\kappa_t=K\Bar{q}(\varepsilon_t/\sqrt{\tau_t}+\varepsilon_t^2/\tau_t + \tau_t)=\mathcal{O}(K\qmax\varepsilon_t^{1/3})$. This gives $\Pr(\mathcal{E}')\geq 1-\mathcal{O}(T^{-1/2})$.

Also, we have $\forall G\in\mathcal{C}'(\hat{F}_t)$ the following by definition 
       \begin{align*}
 \left| \E_G[v|k]-\E_{\hat{F}_t}[v|k]
\right| &\leq \varepsilon_t/\sqrt{\tau_t}+\varepsilon_t^2/\tau_t=\rho_t/\bar{q} \forall \ i\in[I]\backslash\bm{z}_{\tau}.
    \end{align*}
    For $\tilde{x}_{\hat{F}_t}$ using Lemma~\ref{lemma:samplecomplexityimp}, we get 
    \begin{align*}
        \texttt{GAP}(\tilde{x}_{\hat{F}_t},\hat{F}_t,\mathcal{C}'(\hat{F}_t)) \leq K\Bar{q}(\varepsilon_t/\sqrt{\tau_t}+\varepsilon_t^2/\tau_t + \tau_t)=\kappa_t.
    \end{align*}
    
Using DKW we also have
\begin{align*}
    1-\text{Pr}(\mathcal{E})\leq \sum_{t\in[T]}2\exp(-2t\xi_t^2).
\end{align*}
For $\xi_t = \sqrt{\frac{\log(2T^{3/2})}{2t}}$ we get $\text{Pr}(\mathcal{E})\geq 1-\mathcal{O}(T^{-1/2})$ and therefore $\text{Pr}(\mathcal{E}\cap\mathcal{E}')\geq 1-\mathcal{O}(T^{-1/2})$.

\textit{Regret upper bound}:  $F^b_t\in\mathcal{C}(\hat{F}_t)$ by Assumption~1. Under the event $\mathcal{E}\cap\mathcal{E}'$ decompose the $t^{\text{th}}$-regret term as follows
\begin{align*}
    \texttt{OPT}(F^\star) - Rev(\tilde{\bm{x}}_{\hat{F}_t};F^b_t,F^\star)
    &= \underbrace{\bigl[\texttt{OPT}(F^\star)-\texttt{OPT}(\hat{F}_t)\bigr]}_{\mathrm{(A)}}
    +\underbrace{\bigl[\texttt{OPT}(\hat{F}_t)-Rev(\tilde{\bm{x}}_{\hat{F}_t};F^b_t,\hat{F}_t)\bigr]}_{\mathrm{(B)}}\\
    &+\underbrace{\bigl[Rev(\tilde{\bm{x}}_{\hat{F}_t};F^b_t,\hat{F}_t)-Rev(\tilde{\bm{x}}_{\hat{F}_t};F^b_t,F^\star)\bigr]}_{\mathrm{(C)}}.
\end{align*}

 Since $F^b_t\in\mathcal{C}'(\hat{F}_t)$, we get $  \mathrm{(B)}\leq \texttt{GAP}(\tilde{x}_{\hat{F}_t},\hat{F}_t,\mathcal{C}'(\hat{F}_t)) \leq \kappa_t$\\
 For $\mathrm{(C)}$ the buyer's decision is the same, so it can be upper bounded by $\Bar{q}\sup_{v}|F^\star(v)-\hat{F}_t(v)|\sum_{k=1}^Kd_{\mathrm{TV}}(g_k)\leq 2\Bar{q}\xi_tK$ since $g$ has atmost 2 jumps.\\
Define $\tilde{\bm{x}}_{F^\star}=\mathsf{Tr}(\bm{x}_{F^\star};\tau_t,\rho_t)$ for $\bm{x}_{F^\star}=((K^\star,\bm{t}^\star),(p_i^\star,q_i^\star)_{i=1}^K)$ solution of $\texttt{L}_{F^\star}$. Then $\forall G\in\mathcal{C}'(F^\star) $ we have the following by definition
       \begin{align*}
 \left| \E_G[v|k]-\E_{F^\star}[v|k]
\right| &\leq \rho_t/\Bar{q} \quad  \forall \ k\in[K^\star]: t^\star_k-t^\star_{k-1}\geq \tau_t,
    \end{align*}
    and using Lemma~\ref{lemma:samplecomplexityimp}, we get 
    \begin{align*}
        \texttt{GAP}(\tilde{\bm{x}}_{F^\star},F^\star,\mathcal{C}'(F^\star)) \leq \kappa_t.
    \end{align*}
   Recall that under $\mathcal{E}\cap\mathcal{E}'$, we also have $\hat{F}_t\in \mathcal{C}'(F^\star)\ \forall t\in[T]$. 
    For $\mathrm{(A)}$ we perform the following decomposition
\begin{align*}
    \mathrm{(A)}
    &\leq\underbrace{\bigl[\texttt{OPT}(F^\star)-Rev(\tilde{\bm{x}}_{F^\star};\hat{F}_t,F^\star)\bigr]}_{\leq\,\texttt{GAP}(\tilde{\bm{x}}_{F^\star},F^\star,\mathcal{C}'(F^\star))\,\leq\,\kappa_t}
    +\underbrace{\bigl[Rev(\tilde{\bm{x}}_{F^\star};\hat{F}_t,F^\star)-Rev(\tilde{\bm{x}}_{F^\star};\hat{F}_t,\hat{F}_t)\bigr]}_{\leq 2\Bar{q}\xi_tK\ \text{ using same TV argument as (C)}}\\
    &+\underbrace{\bigl[Rev(\tilde{\bm{x}}_{F^\star};\hat{F}_t,\hat{F}_t)-\texttt{OPT}(\hat{F}_t)\bigr]}_{\leq\,0\ 
    \text{due to suboptimality of }\tilde{\bm{x}}_{F^\star}\text{ under }\hat{F}_t}\\
    &\leq\kappa_t+2\Bar{q}\xi_tK.
\end{align*}
 Combining we get overall 
\begin{align*}
    \texttt{OPT}(F^\star)-Rev(\tilde{\bm{x}}_{\hat{F}_t};F^b_t,F^\star)\leq 2\kappa_t+4\Bar{q}\xi_tK.
\end{align*}

Finally, define $Y_t:=\texttt{rev}_t(\tilde{x}_{\hat{F}_t};F^b_t)-Rev(\tilde{\bm{x}}_{\hat{F}_t};F^b_t,F^\star)$. Since $\tilde{\bm{x}}_{\hat{F}_t}$ and $F^b_t$ are $\mathcal{H}_t$-measurable and $v_t|\mathcal{H}_t\sim F^\star$, we have $\E[Y_t|\mathcal{H}_t]=0$ and $|Y_t|\leq \Bar{q}$ where $\mathcal{H}_t$ denotes the $\sigma$-algebra of all the histories upto time $t$.  By Azuma-Hoeffding, we get
\begin{align*}
    \text{Pr}\left(\sum_{t=1}^T(-Y_t)\geq z\Bar{q}\right)\leq\exp\!\left(-\frac{z^2}{2T}\right).
\end{align*}
Setting $z=\sqrt{T\log T}$ gives $T^{-1/2}$ upper bound.

Thus overall on the event $\mathcal{E}\cap\mathcal{E}'\cap\{\sum_t(-Y_t)\leq z\Bar{q}\}$, which now holds with probability $\geq 1-\mathcal{O}(T^{-1/2})$, we get 
\begin{align*}
    \texttt{Reg}
    =\sum_{t=1}^T\bigl(\texttt{OPT}(F^\star)-\texttt{rev}_t\bigr)
    &=\sum_{t=1}^T\bigl(\texttt{OPT}(F^\star)-Rev(\tilde{\bm{x}}_{\hat{F}_t};F^b_t,F^\star)\bigr)+\sum_{t=1}^T(-Y_t)\\
    &\leq\sum_{t=1}^T(2\kappa_t+4\Bar{q}\xi_tK)+\Bar{q}\sqrt{T\log T}.
\end{align*}

Substituting $\tau_t=\varepsilon_t^{1/3}$ gives us $\kappa_t =K\Bar{q}(\varepsilon_t/\sqrt{\tau_t}+\varepsilon_t^2/\tau_t + \tau_t) = \mathcal{O}(K\bar q\varepsilon_t^{1/3})$ and since $\varepsilon_t=\sqrt{5C\log T/(2t)}$, we have $\varepsilon_t^{1/3}=(5C\log T/(2t))^{2/3}$ and $\sum_{t=1}^T\xi_t=\mathcal{O}
(\sqrt{3T\log T})$, which gives us 
\begin{align*}
    \texttt{Reg}\leq \mathcal{O}\left(K\Bar{q}(\log T)^{1/3}T^{2/3}\right)
\end{align*}
\end{proof}

\section{Proofs for the FPTAS}\label{sec:appfptas}
\begin{lemma}
    Consider any quantile interval $[t_1,t_2]$ with $0\leq t_1< t_2\leq 1$, and define $
        \tilde t_1 := \varepsilon \lfloor \tfrac{t_1}{\varepsilon}\rfloor,
        \tilde t_2 := \varepsilon \lfloor \tfrac{t_2}{\varepsilon}\rfloor.$
    Then $  |(t_2-t_1)-(\tilde t_2-\tilde t_1)|\leq \varepsilon.$
    Moreover, let $\mu(t_1,t_2) = \dfrac{\int_{t_1}^{t_2}Q(u)\mathrm{d}u}{t_2-t_1}$ and define $\mu(\tilde{t}_1,\tilde{t}_2)=\dfrac{\int_{\tilde{t}_1}^{\tilde{t}_2}Q(u)\mathrm{d}u}{\tilde{t}_2-\tilde{t}_1}$ and as $Q(\tilde{t}_1)$ whenever $\tilde{t}_1=\tilde{t}_2$. Then 
    \begin{align*}
         |\mu(t_1,t_2)
        -
        \mu(\tilde{t}_1,\tilde{t}_2)|\leq  \frac{3 \varepsilon}{t_2-t_1}.
    \end{align*}
\label{lemma:posteriorerrorquantilegrid}
\end{lemma}

\begin{proof}[Proof of Lemma~\ref{lemma:posteriorerrorquantilegrid}]
By definition, $\tilde{t}_1\leq t_1$ and $t_1-\varepsilon\leq \tilde{t}_1$ and similarly, $\tilde{t}_2\leq t_2$ and $t_2-\varepsilon\leq \tilde{t}_2$. Thus $t_{2}-t_{1}\leq \tilde{t}_2+\varepsilon-\tilde{t}_{1}$ and $t_2-t_{1}\geq \tilde{t}_2-\tilde{t}_{1}-\varepsilon$. Thus $|t_2-t_{1}-(\tilde{t}_2-\tilde{t}_{1})|\leq \varepsilon$. \\
Thus if $\tilde{t}_2-\tilde{t}_1=0$, then $t_2-t_1<\varepsilon$, then the inequality immediately follows
\begin{align*}
    |\mu(t_2,t_1)-Q(\tilde{t}_1)|\leq 1\leq 3 \varepsilon/(t_2-t_1).
\end{align*}
Otherwise for $\tilde{t}_2>\tilde{t}_1$ we have
\begin{align*}
    \left|\int_{t_{1}}^{t_2} Q(u)\mathrm{d}u -\int_{\tilde{t}_{1}}^{\tilde{t}_2} Q(u)\mathrm{d}u\right| &\leq  \left|\int_{t_{1}}^{\tilde{t}_{1}} Q(u)\mathrm{d}u\right| +\left|\int_{t_2}^{\tilde{t}_2} Q(u)\mathrm{d}u\right|\leq  2\varepsilon
\end{align*}
which gives
\begin{align*}
    |\mu(t_1,t_2) -\mu(\tilde{t}_1,\tilde{t}_1)|&\leq \dfrac{\left|\int_{t_{1}}^{t_2} Q(u)\mathrm{d}u -\int_{\tilde{t}_{1}}^{\tilde{t}_2} Q(u)\mathrm{d}u\right|}{t_2-t_{1}}+\int_{\tilde{t}_{1}}^{\tilde{t}_2} Q(u)\mathrm{d}u\left|\dfrac{1}{t_2-t_{1}}-\dfrac{1}{\tilde{t}_2-\tilde{t}_{1}}\right|\\
    &\leq \dfrac{2 \varepsilon}{t_1-t_{2}}+\int_{\tilde{t}_{1}}^{\tilde{t}_2} Q(u)\mathrm{d}u\left|\dfrac{1}{t_1-t_{2}}-\dfrac{1}{\tilde{t}_1-\tilde{t}_{2}}\right|\\
     &\leq \dfrac{2 \varepsilon}{t_2-t_{1}}+ (\tilde{t}_2-\tilde{t}_{1}) \dfrac{\varepsilon}{(t_2-t_{1})(\tilde{t}_2-\tilde{t}_{1})}\\
     &\leq 3 \varepsilon/\tau
\end{align*}

\end{proof}

\begin{theorem}[Correctness of Algorithm~\ref{alg:dp}]
For $\texttt{L}^s_F$ restricted to quantiles in $\mathcal{E}_t$ and qualities in $\mathcal{E}_q$, the 
Algorithm~\ref{alg:dp} outputs an optimal menu $\{(p_k, q_{j_k})\}_{k=1}^{k^\star}$  and  a valid information structure $\mathcal{I}=(k^\star,\bm{t})$.
\label{thm:dp_correctness}
\end{theorem}

\begin{proof}
    For $k\in [K]$, $r\in [n_t]$ and $j\in[n_{q}]$, define 
    \begin{align*}
        V(k,r,j) := \max\left\{\sum_{i=1}^k\mathcal{U}[\ell_{i-1},\ell_i,j_{i-1},j_i] \Big| 0=\ell_0<\ldots<\ell_k=r, 0=j_0\right\}
    \end{align*}
which denotes the maximum revenue that can be obtain from partitioning $(0,t_r]$ into $k$ nonempty quantile grid intervals with the last interval assigned quality $q_j$. Define $V(0,0,0)=0$ and $V(0,r,j)=-\infty \ \forall (r,j)\neq (0,0)$. We claim $\texttt{Rev}[k,r,j]=V(k,r,j)\ \forall\ k,r,j$. \\
We will show this by an induction over $k$.\\
\textit{Base case} is when $k=0$, by initialization we have equality.\\
\textit{Inductive step} assume that $\texttt{Rev}[k-1,\cdot,\cdot]=V(k-1,\cdot,\cdot).$ The \texttt{DP} recursion gives us 
\begin{align*}
    \texttt{Rev}[k,r,j] = \max_{\substack{0\leq \ell<r\\0\leq j'\leq j}}\{\texttt{Rev}[k-1,\ell,j']+\mathcal{U}[\ell,r,j',j]\} = \max_{\substack{0\leq \ell<r\\0\leq j'\leq j}}\{V(k-1,\ell,j')+\mathcal{U}[\ell,r,j',j]\}
\end{align*}
We will first show that $\texttt{Rev}[k,r,j]\leq V(k,r,j)\ \forall r,j$. To do this, fix any $0\leq \ell<r$ and $0\leq j'\leq j$ and let $0=\ell_0<\ldots <\ell_{k-1}=\ell$ be a partition of intervals $(k-1)$ with quality assignment $0=j_0\leq \ldots\leq j_{k-1}=j'$ where the quantile interval indexed by $[\ell_{k-2},\ell_{k-1}]$ is assigned quality indexed $j_{k-1}$ and so on. Appending $(\ell,r]$ at the end of the quantile interval gives us a valid valid partition as $\ell<r$ and assigning quality $q_j$ to this interval is valid since it maintains monotonicity, such a scheme gives us revenue of $V(k-1,\ell,j')+\mathcal{U}[\ell,r,j',j]$. By feasibility we have $V(k-1,\ell,j')+\mathcal{U}[\ell,r,j',j]\leq V(k,r,j)\ \forall 0\leq \ell<r,0\leq j'\leq j$. Taking maximum on all such $\ell,j'$ gives us $\texttt{Rev}[k,r,j]\leq V(k,r,j) \forall r,j$.\\

 Now we show that $V(k,r,j)\leq \texttt{Rev}[k,r,j]\ \forall r,j$. Consider the $k$ partition $0\leq \ell_0<\ldots<\ell_k=r$ and the quality assignment of $0=j_0\leq \ldots\leq j_k=j$ that obtains the value $V(k,r,j)$. Set $\ell=\ell_{k-1}$ and $j'=j_{k-1}$. Then the partition $0=\ell_0<\ldots<\ell$ is a valid  $(k-1)$ partition with feasible quality assingment $0=j_0\leq \ldots \leq j'$ giving us $\sum_{i=1}^{k-1}\mathcal{U}[\ell_{i-1},\ell_i,j_{i-1},j_i] \leq V(k-1,\ell,j')$. Finally, by definition, we have 
\begin{align*}
    V(k,r,j) = \sum_{i=1}^{k-1}\mathcal{U}[\ell_{i-1},\ell_i,j_{i-1},j_i] + \mathcal{U}[\ell,r,j',j] \leq V(k-1,\ell,j')+\mathcal{U}[\ell,r,j',j]\leq \texttt{Rev}[k,r,j]
\end{align*}
where the last inequality follows from the \texttt{DP} update.

Optimality of the partitions and quality assignment follows from the greedy backward pass, and optimality  and \texttt{IC} of the menu follow from the payment formula. 
    
\end{proof}

\begin{theorem}
    Let $\bm{x} = \left((I,\bm{t}),(p_i,q_i)_{i\in[I]}\right)$ be the solution of $\texttt{L}_F$. Then the \texttt{DP} output $\bm{x}' = \left((I',\bm{t}'),(p'_i,q'_i)_{i\in[I]}\right)$ is such that $\texttt{Rev}(\bm{x}';F)\geq \texttt{Rev}(\bm{x};F)- \varepsilon_{\mathsf{OPT}}$ where the optimization error is given as  $\varepsilon_{\mathsf{OPT}}=I(6\qmax +1+\qmax)\sqrt{\varepsilon} $.
    \label{thm:errorduetogrids}
\end{theorem}

\begin{proof}[Proof of Theorem~\ref{thm:errorduetogrids}]
    We use the following decomposition
    \begin{align*}
        Rev(\bm{x};F) - Rev(\bm{x}';F) =  Rev(\bm{x};F) -  Rev(\bm{\tilde{x}};F) +\underbrace{{Rev(\bm{\tilde{x}};F)-Rev(\bm{x}';F)}}_{\leq 0\ \text{(by correctness of \texttt{DP})}}
    \end{align*}
    where $\bm{\tilde{x}}\gets \mathsf{modification}(\bm{x})$ such that for $\bm{\tilde{x}}=((I,\tilde{\bm{t}}),(\tilde{p}_i,\tilde{q}_i)_{i\in[I]})$ where each component of $\tilde{\bm{t}}$ and each $\tilde{q}_i$ lie on the grid $\mathcal{E}_{t}$ and $\mathcal{E}_{q}$ respectively, with $\tilde{p}_i$ defined later. In particular let $\tilde{q}_i=\varepsilon \lfloor \frac{q_i}{\varepsilon}\rfloor$ and $\tilde{t}_i=\varepsilon \lfloor \frac{t_i}{\varepsilon}\rfloor$. Thus the monotonicity of $\tilde{q}_i$ is preserved.  Based on the modification, we have $\tilde{q}_i\leq q_i$ and $q_i-\varepsilon\leq \tilde{q}_i$ and similarly, $\tilde{t}_i\leq t_i$ and $t_i-\varepsilon\leq \tilde{t}_i$.
    Let $\bm{z}_{\tau}=\{i\in [I]:t_{i-1}-t_i<\tau\}$.\\
    Consider $i\not\in \bm{z}_{\tau}$. Then using Lemma~\ref{lemma:posteriorerrorquantilegrid} we get
    \begin{align*}
        \tilde{q}_i\cdot   \mu(\tilde{t}_{i-1},\tilde{t}_i) - p_i&\geq q_i\cdot   \mu(t_{i-1},t_i) - p_i -3\qmax   \varepsilon/(t_{i-1}-t_i)-\varepsilon \\
        &\geq q_i\cdot   \mu(t_{i-1},t_i) - p_i -3\qmax   \varepsilon/\tau-\varepsilon \\
        &\geq \max\{0,\max_j\{q_j\cdot \mu(t_{i-1},t_i)-p_j\}\}-3\qmax   \varepsilon/\tau-\varepsilon \\
        &\geq \max\{0,\max_j\{q_j\cdot \mu(\tilde{t}_{i-1},\tilde{t}_i)-p_j\}\}-6\qmax   \varepsilon/\tau-\varepsilon \\
         &\geq \max\{0,\max_j\{\tilde{q}_j\cdot \mu(\tilde{t}_{i-1},\tilde{t}_i)-p_j\}\}-6\qmax   \varepsilon/\tau-\varepsilon 
    \end{align*}
    Let $\tilde{p}_i = p_i-i\rho$ for $\rho=6\qmax   \varepsilon/\tau+\varepsilon $. Then observe that for all $j<i$
    \begin{align*}
        \mu(\tilde{t}_{i-1},\tilde{t}_i)\cdot \tilde{q}_i-\tilde{p}_i&=    \mu(\tilde{t}_{i-1},\tilde{t}_i)\cdot \tilde{q}_i-p_i+i\rho \\
        &\geq \mu(\tilde{t}_{i-1},\tilde{t}_i)\cdot \tilde{q}_j - p_j+i\rho - \rho\\
        &= \mu(\tilde{t}_{i-1},\tilde{t}_i)\cdot \tilde{q}_j-\tilde{p}_j+(i-j)\rho-\rho\\
        &\geq\mu(\tilde{t}_{i-1},\tilde{t}_i)\cdot \tilde{q}_j - \tilde{p}_j
    \end{align*}
    Therefore, buyer upon receiving signal $i\not\in\bm{z}_{\tau}$ agrees to pay at least $p_i-I\rho$, generating a margin of $p_i-c(q_i)-I\rho$, therefore, 
    \begin{align*}
        Rev(\bm{x};F) -  Rev(\bm{\tilde{x}};F)\leq I(\rho+\qmax\tau)
    \end{align*}
    where for upper bound we assume the worst case that buyers with signal $i\in\bm{z}_{\tau}$ do not purchase anything.  Setting $\tau = \sqrt{\varepsilon}$, we get 
    \begin{align*}
          Rev(\bm{x};F) - Rev(\bm{x}';F)\leq I(6\qmax +1+\qmax)\sqrt{\varepsilon} 
    \end{align*}

 \end{proof}

\begin{theorem}
  For discrete $F$ with $\mathcal{O}(m)$ atoms,  Algorithm~\ref{alg:dp} has a runtime complexity of  $\mathcal{O}\left(K\cdot m^2\cdot n_q^2\right)$ which is $\mathcal{O}(1/\varepsilon^4)$ for $m=\mathcal{O}(1/\varepsilon)$, and has memory complexity $\mathcal{O}(K\cdot m\cdot n_q)$.
    \label{thm:complexities}
\end{theorem}
\begin{proof}[Proof of Thm~\ref{thm:complexities}]
    \textit{Complexity.}
Algorithm~\ref{alg:dp} preprocesses $A(\cdot)$ at all $m+1$ grid points in $\mathcal{O}(m+N)$ time, so that each $\mu_F(\ell,r)$ is evaluated in $\mathcal{O}(1)$.  The DP table has $K\cdot m\cdot n_q$ entries where filling each entry requires maximizing over $O(m\cdot n_q)$ pairs $(\ell,j')$, giving a total forward-pass runtime complexity of $\mathcal{O}\left(K\cdot m^2\cdot n_q^2\right)$.
The backward pass runs in $ O(K)$ with price computation in $O(K^2)$, and the signaling scheme in $O(K\cdot N)$, all dominated by the forward pass.  The grid $\{t_0,\ldots,t_m\}$ can be chosen freely, choosing $m$ points with maximum spacing $\Delta:=\max_i(t_i-t_{i-1})=\mathcal{O}(\varepsilon)$ requires $m=\mathcal{O}(1/\varepsilon)$.  With $n_q=\mathcal{O}((\qmax-\qmin)/\varepsilon)$, the total time is thus $\mathcal{O}(K/\varepsilon^4)$ which polynomial in $1/\varepsilon$ for fixed $K$, $\qmax$, and $\qmin$.

 Memory complexity is $\mathcal{O}(K\cdot m\cdot n_q)$ for the \texttt{Rev} and \texttt{ptr} tables. 
\end{proof}

\begin{theorem}[{Additive-FPTAS}] For any $\varepsilon_{\mathsf{OPT}}>0$, given a discrete distribution with support size $N$, \Cref{alg:dp} runs in $\mathsf{poly}(K,N,\overline{q},1/\varepsilon_{\mathsf{OPT}})$ time and outputs a menu $\{(p_k, q_{j_k})\}_{k=1}^{k^\star}$  and
    a monotone paritional signaling scheme $\mathcal{I}=(k^\star,(t_{r_k}:k\in [k^\star]))$ with revenue exceeding ${Rev}^\star(F)-\varepsilon_{\mathsf{OPT}}$ by setting the resolution of the grid as $\varepsilon = (\varepsilon_{\mathsf{OPT}})^2/(K(6\qmax +1+\qmax))^2$.
\end{theorem}

\begin{proof}
 The proof follows by a combination of earlier results, namely, Theorems~ \ref{thm:dp_correctness}, \ref{thm:errorduetogrids}, and \ref{thm:complexities}.
\end{proof}

\section{Equivalent Representations of Signaling Schemes}\label{app:equiv}
\begin{lemma}
    Consider the following equivalent representation of the signaling scheme.
    \begin{enumerate}[itemindent=1.3em, leftmargin=0em, itemsep=0em]
        \item Quantile space representation: $\mathcal{S} = \{r_i\}_{i=1}^{[I]}$ with the signaling scheme as\\ 
        $\pi(r_i|u) = \bm{1}\{u\in [t_{i-1},t_i]\}$ for all $u\sim \texttt{Unif}[0,1]$ where $t_0=1-t_I=0$ with  $t_{i-1}\leq t_i$ and $t_i\in [0,1]\ \forall\ i\in[I]$. 
        \item Value space representation: $\mathcal{S} = \{s_i\}_{i=1}^{[I]}$ $w_0\leq w_1\leq\cdots\leq w_I$, $\xi_i\in[0,1]$ with $\xi_0=0$, $\xi_I=1$, and $\xi_{i-1}\leq\xi_i$ whenever $w_{i-1}=w_i$, with the signaling scheme as
        \begin{align*}
            \pi(s_i|v) = \begin{cases}
                (1-\xi_{i-1})\mathbf{1}\{v=w_{i-1}\}+\mathbf{1}\{w_{i-1}<v<w_i\}+\xi_i\mathbf{1}\{v=w_i\} & \text{if}\ w_{i-1}<w_i,\\
                (\xi_i-\xi_{i-1})\mathbf{1}\{v=w_{i-1}\} & \text{if}\ w_{i-1}=w_i.
            \end{cases}
        \end{align*}
        for all $v\sim F$.
    \end{enumerate}
    Then for $t_i = (1-\xi_i)F(w_i^-)+\xi_iF(w_i)$, the posterior mean under signal $s_i$ and signal $r_i$ is equal to $\int_{t_{i-1}}^{t_i}Q(t)\mathrm{d}t/(t_i-t_{i-1})$ where $Q(t) = \inf\{v:F(v)\geq t\}$. \label{lemma:posteriormeanequivalence}
\end{lemma}

\begin{proof}[Proof of Lemma~\ref{lemma:posteriormeanequivalence}]
In the value space, it is a little messy and can be written as $\pi(r_i|v) = \lambda([F(v^-),F(v)]\cap [t_{i-1},t_i])/(F(v)-F(v^-))$, where $\lambda(\cdot)$ captures the length of the interval, concretely, a Lebesgue measure. The posterior expected value under $r_i$ is $\texttt{N}/\texttt{D}$ where
    \begin{align*}
        \texttt{N} = \int_{0}^1 Q(u)\bm{1}\{u\in[t_{i-1},t_i]\} \mathrm{d}u = \int_{t_{i-1}}^{t_i} Q(u)\mathrm{d}u
    \end{align*} where we used $v=Q(u)$ and
    \begin{align*}
        \texttt{D}=\int_{0}^1 \bm{1}\{u\in[t_{i-1},t_i]\} \mathrm{d}u =t_{i}-t_{i-1}
    \end{align*}

Consider first the case when $w_{i-1}<w_i$, this implies $F$ is continuous and strictly increasing on $(w_{i-1},w_i)$.  Then the posterior expected value under $s_i$ is $\texttt{N}'/\texttt{D}'$ where
        \begin{align*}
            \texttt{N}' = \int_{0}^{\vmax} v\mathrm{d}F(v) = w_{i-1} (1-\xi_{i-1})(F(w_{i-1})-F(w_{i-1}^-))+\int_{(w_{i-1},w_i)} v\mathrm{d}F(v) + w_i\xi_i(F(w_{i})-F(w_{i}^-)) \end{align*}
            Note by the definition of $Q(\cdot)$ we have $Q(t) = w_{i-1}\ \forall t\in [t_{i-1},F(w_{i-1})]$ and $Q(t) = w_i\ \forall t\in[F(w_i^-),t_i]$. Therefore $\int_{(w_{i-1},w_i)}v\mathrm{d}F(v) = \int_{F(w_{i-1})}^{F(w_i^-)}Q(t)\mathrm{d}t$, substituting this in $\texttt{N}'$ and writing $\xi_i$ in terms of $t_i$ gives us $\texttt{N}'=\texttt{N}$. Also,
            \begin{align*}
                \texttt{D}' = \int_{0}^{\bar{v}}\mathrm{d}F(v)&=  (1-\xi_{i-1})(F(w_{i-1})-F(w_{i-1}^-))+\int_{(w_{i-1},w_i)} \mathrm{d}F(v) + \xi_i(F(w_{i})-F(w_{i}^-)) \\
                &=   (1-\xi_{i-1})(F(w_{i-1})-F(w_{i-1}^-))+(F(w_{i}^-)-F(w_{i-1})+ \xi_i(F(w_{i})-F(w_{i}^-)) \\
                &= t_{i}-t_{i-1}=\texttt{D}
            \end{align*}
        For the case when $w_{i-1}=w_i, t_{i-1},t_i\in [F(w_i^-),F(w)]$ thus $Q(t) = w_i\ \forall t\in[t_{i-1},t_i]$. Thus $\int_{0}^{\bar{v}}v\mathrm{d}F(v) = w_i(\xi_i-\xi_{i-1})(F(w_i)-F(w_i^-))$ and $\int_0^{\bar{v}}\mathrm{d}F(v) = (\xi_i-\xi_{i-1})(F(w_i)-F(w_i^-))$, which writing $\xi_i$ in terms of $t_i$ gives the desired result.
\end{proof}
